\documentclass[letter,11pt,final]{article}
\usepackage{typearea}
\typearea{14}

\usepackage{typearea}
\typearea{14}

\usepackage{authblk}
\usepackage[numbers,sort]{natbib}
\usepackage{comment}
\usepackage[T1]{fontenc}
\usepackage[noend]{algorithmic}
\usepackage{algorithm}
\usepackage{multirow}
\usepackage{amssymb}

\usepackage{amsmath}

\usepackage{amsthm,thm-restate}
\usepackage{todonotes}
\usepackage{enumerate}
\usepackage[shortlabels]{enumitem}
\usepackage{dsfont}
\usepackage[colorlinks=true,citecolor=blue,urlcolor=blue,bookmarks=false]{hyperref}
\usepackage{soul}  
\usepackage{cleveref} 
\usepackage{footnote}
\usepackage{subcaption}
\usepackage{xspace}
\usepackage{nicefrac}
\usepackage{tikz}
\usetikzlibrary{decorations.pathmorphing, positioning, decorations.markings, decorations.pathreplacing} 
\usetikzlibrary[topaths]
\usepackage{ragged2e} 
\tikzstyle{replaced edge} = [line width=1.5pt,blue!80]
\tikzstyle{vertex}= [draw,circle,inner sep=0.1cm]
\tikzstyle{BigLetterVertex}= [draw,circle,inner sep=0.075cm]
\usepackage{framed}
\usepackage[framemethod=tikz]{mdframed}


\theoremstyle{plain}
\newtheorem{theorem}{Theorem}[section]

\newtheorem{lemma}[theorem]{Lemma}
\newtheorem{claim}[theorem]{Claim}
\newtheorem{definition}[theorem]{Definition}
\newtheorem{observation}[theorem]{Observation}

\DeclareMathOperator{\poly}{poly}

\newcommand{\eps}{\varepsilon}

\newcommand{\calS}{\mathcal{S}}
\newcommand{\calA}{\mathcal{A}}
\newcommand{\altA}{\mathbb{A}}

\newcommand{\calC}{\mathcal{C}}
\newcommand{\calP}{\mathcal{P}}
\newcommand{\calF}{\mathcal{F}}

\newcommand{\calM}{\mathcal{M}}
\newcommand{\calD}{\mathcal{D}}

\newcommand{\pth}{\mathrm{p}}
\newcommand{\dist}{\mathrm{dist}}
\newcommand{\clip}{\mathrm{clip}}
\newcommand{\ps}{\mathrm{ps}}
\newcommand{\pt}{\mathrm{pt}}
\newcommand{\disj}{\mathrm{disj}}

\newcommand{\tO}{\tilde{O}}
\newcommand{\tI}{\tilde{I}}
\newcommand{\tOmega}{\tilde{\Omega}}
\newcommand{\mc}[1]{\mathcal{#1}}
\newcommand{\mr}[1]{\mathrm{#1}}

\newcommand{\E}{\mathbb{E}}

\newcommand{\opt}{\mathrm{opt}}

\renewcommand{\cong}{\mathrm{cong}}

\newcommand{\at}[1]{^{(#1)}} 

\newcommand{\wcsub}[1]{\ensuremath{\textsc{WC\-sub\-net\-work}(#1)}}
\newcommand{\movingcut}[1]{\textsc{Moving\-Cut}\allowbreak(#1)}
\newcommand{\routing}[1]{\textsc{Routing}\allowbreak(#1)}
\newcommand{\comm}[1]{\textsc{Communicating}\allowbreak(#1)}

\newcommand{\shortcut}[1]{\textsc{Shortcut\-Quality}\allowbreak(#1)}
\newcommand{\pairshortcut}[1]{\textsc{Shortcut\-Quality}_2\allowbreak(#1)}

\newenvironment{wrapper}[1]
{
	\begin{center}
		\begin{minipage}{\linewidth}
			\begin{mdframed}[hidealllines=true, backgroundcolor=gray!20, leftmargin=0cm,innerleftmargin=0.5cm,innerrightmargin=0.5cm,innertopmargin=0.5cm,innerbottommargin=0.5cm,roundcorner=10pt]
				#1}
			{\end{mdframed}
		\end{minipage}
	\end{center}
} 

\title{Universally-Optimal Distributed Algorithms for Known Topologies}
\author[1]{Bernhard Haeupler}
\author[2]{David Wajc}
\author[3]{Goran Zuzic}
\affil[1]{Carnegie Mellon University}
\affil[2]{Stanford University}
\affil[3]{ETH Zurich}
\date{\vspace{-1cm}}

\begin{document}
\maketitle

\pagenumbering{gobble}

\begin{abstract}
  Many distributed optimization algorithms achieve \emph{existentially-optimal} running times, meaning that there exists \emph{some} pathological worst-case topology on which no algorithm can do better. Still, most networks of interest allow for exponentially faster algorithms. This motivates two questions:

  \begin{itemize}
  \item What network topology parameters determine the complexity of distributed optimization? 
  \item Are there \emph{universally-optimal} algorithms that are as fast as possible on \emph{every} topology?
  \end{itemize} 

\smallskip

We resolve these 25-year-old open problems in the known-topology setting (i.e., supported CONGEST) for a wide class of global network optimization problems including MST, $(1+\eps)$-min cut, various approximate shortest paths problems, sub-graph connectivity, etc.
 
\smallskip

In particular, we provide several (equivalent) graph parameters and show they are tight \emph{universal lower bounds} for the above problems, fully characterizing their inherent complexity. Our results also imply that algorithms based on the low-congestion shortcut framework match the above lower bound, making them universally optimal if shortcuts are efficiently approximable. We leverage a recent result in hop-constrained oblivious routing to show this is the case if the topology is known---giving universally-optimal algorithms for all above problems. 

\end{abstract}


	\newpage
	
	\pagenumbering{arabic}


\tableofcontents
\section{Introduction}

Much of modern large-scale graph processing and network analysis is done using systems like Google's Pregel~\cite{malewicz2010pregel}, Facebook's Giraph~\cite{han2015giraph,ching2015one}, or Apache's Spark GraphX~\cite{gonzalez2014graphx}. These systems implement synchronous message-passing algorithms, in which nodes send (small) messages to their neighbors in each round.\footnote{For the sake of concreteness, we limit message sizes to $O(\log n)$ bits, where $n$ is the number of network nodes. This is exactly the classic CONGEST model of distributed computation~\citep{peleg2000distributed}, or the supported CONGEST model~\cite{schmid2013exploiting}, if the network is known and preprocessing is allowed. In what follows, we use the terms topology and network graph interchangeably.} 

\smallskip

This has motivated a recent, broad, concentrated, and highly-successful effort to advance our theoretical understanding of such algorithms for fundamental network optimization problems, such as minimum-spanning trees (MST)~\cite{khan2008fast,elkin2017simple,kutten1998fast,pandurangan2017time,elkin2006faster}, shortest paths~\cite{nanongkai2014distributed,Frischknecht-Diameter-2012,Holzer-Paths-2012,Lenzen:2013,Lenzen:2015, henzinger16almost, elkin2017distributed, huang2017distributed}, flows~\cite{ghaffari2015near}, and cuts~\cite{nanongkai2014almost, ghaffari2016distributed,dory2020distributed}. As a result, many fundamental optimization problems now have worst-case-optimal CONGEST algorithms, running in $\tilde{\Theta}(\sqrt{n} + D)$ rounds on every $n$-node network with diameter $D$.\footnote{Throughout, we use $\tilde{O},\tilde{\Omega}$ and $\tilde{\Theta}$ to suppress $\poly\log n$ terms. E.g., $\tilde{O}(f(n)) = O(f(n) \log^{O(1)} n)$.} In general, these running times cannot be improved due to unconditional lower bounds \cite{peleg2000near,dassarma2012distributed,elkin2006unconditional} showing that there \emph{exist} pathological $n$-node topologies with small diameter on which any non-trivial optimization problem requires $\tilde{\Omega}(\sqrt{n})$ rounds. Fittingly, this type of worst-case optimality is also called \emph{existential optimality}.

\smallskip

While these results are remarkable achievements, this paper emphatically argues that one cannot stop with such worst-case optimal algorithms. In particular, existential optimality says nothing about  the performance of an algorithm compared to what is achievable on real-world networks, which are never worst-case, and might allow for drastically faster running times compared to a pathological worst-case instance. For example, it is a well-established fact that essentially all real-world topologies have network diameters that are very small compared to the network size~\cite{newman2003structure} (this is known as the small-world effect), and essentially all mathematical models for practical networks feature diameters that are at most polylogarithmic in $n$. 
Despite this fact, for such networks the  existentially-optimal algorithms take $\tilde{\Theta}(\sqrt{n})$ rounds, which is neither practically relevant, nor does it correspond to any observed practical barrier or bottleneck.
This has motivated a concentrated effort in recent years to provide improved algorithms for families of networks of interest
\cite{haeupler2016low,haeupler2016near,ghaffari2018new,haeupler2018minor,haeupler2018faster,ghaffari2017near,kitamura2019low,ghaffari2016distributed,haeupler2018round,ghaffari2017distributed}.
Nonetheless, many practical networks do not fit any of the above families, and so no practically useful algorithm is known for these networks, even if we allow for preprocessing of a known network topology. This strongly motivates a broader search for algorithms which adjust to non-worst-case topologies.

\subsection{When Optimal is not Good Enough: Universal Optimality}
\label{sec:intro-better-optimality}

The search for algorithms with beyond-worst-case guarantees is not new. Indeed, it goes back at least as far as 25 years ago, to an influential paper of \citet*{garay1993sublinear}. 
In that work, Garay et al.~improved upon the existentially-optimal $O(n)$ round minimum spanning tree (MST) algorithm of \citet{awerbuch1987optimal}.\footnote{Note that there exists a pathological worst-case topology requiring $\Omega(n)$ rounds to solve MST: an $n$-node ring graph.}
In particular, they gave an $\tilde{O}(n^{0.613} + D)$ round algorithm. This was in turn improved to an $\tilde{O}(\sqrt{n} + D)$ round algorithm that is existentially-optimal in $n$ and $D$ by \citet{kutten1998fast}.
These latter two papers started and majorly shaped the area of distributed optimization algorithms. Garay, Kutten, and Peleg~informally introduced the concept of universal-optimality as the ultimate guarantee for adjusting to non-worst-case topologies:

\begin{wrapper}
\begin{quote}
  This type of optimality may be thought of as ``existential'' optimality; namely, there are points in the class of input instances under consideration for which the algorithm is optimal. A stronger type of optimality, which we may analogously call “universal” optimality, occurs when the proposed algorithm solves the problem optimally on \emph{every} instance. [...] The interesting question that arises is, therefore, whether it is possible to identify the \emph{inherent graph parameters} associated with the distributed complexity of various fundamental network problems, and develop \emph{universally-optimal} algorithms.~\cite{garay1993sublinear}
\end{quote}
\end{wrapper}

However, formalizing this concept is not as straightforward as it seems---two different readings of the above quote allows for a variety of subtly-but-crucially different formal definitions. In this paper we provide two sensible definitions. (See \Cref{sec:universal-optimality} for a fully formal treatment.)

\smallskip

An algorithm $\calA$ is \emph{instance optimal} if, every instance (i.e., a network $G$ and problem-specific input), its runtime is $\tilde{O}(1)$-competitive with every other always-correct algorithm, including the fastest algorithm for that instance. While interesting and useful in more restrictive contexts, we show that instance optimality is provably unachievable for CONGEST problems like MST.

\smallskip

Next, we say an algorithm $\calA$ is \emph{universally optimal} if, for every network $G$, worst-case runtime of $\calA$ over all inputs on $G$ is $\tilde{O}(1)$-competitive with the worst-case runtime of any other always-correct algorithm, particularly the fastest algorithm for $G$. Equivalently, this definition (implicitly) asks about the inherent graph parameter $X_{\Pi}(G)$ such that there is a correct algorithm for a problem $\Pi$ running in $\tilde{O}(X_{\Pi}(G))$ rounds on $G$, and any correct algorithm for $\Pi$ requires $\tilde{\Omega}(X_{\Pi}(G))$ rounds (on some input). While it is not clear whether there exists a single algorithm with such a property for any non-trivial task, we prove such an algorithm does exist for many distributed problems.

\smallskip 

Note that $X_{\Pi}$ is problem-specific by definition, since different problems $\Pi$ could be characterized by different graph parameters. Remarkably, we show that all the problem studied in this paper form a ``universal complexity class'', and share a common parameter which captures their complexity.

\subsection{Our Results}

We resolve both questions of \citet*{garay1993sublinear} in the supported CONGEST model: we identify graph parameters that fully characterize the inherent complexity of distributed global network optimization, and prove the existence of universally-optimal distributed algorithms that match those parameters.

\smallskip

Specifically, we show that a graph parameter $\shortcut{G}$ is a universal CONGEST lower bound for many important distributed optimization problems---meaning that every correct algorithm requires at least $\tilde{\Omega}(\shortcut{G})$ rounds to compute the output on any network $G$.
The main challenge in obtaining such a result is the very rich space of possible algorithms that might produce correct results on some specially-crafted network $G$, possibly outperforming shortcut-based algorithms. To show this is not possible, we utilize prior work on network coding gaps that relate general algorithms for certain simple communication problems to shortcut quality~\cite{haeupler2020network}. We then combine this result with a combinatorial construction that shows how to find good-quality shortcuts if there exist fast distributed algorithms for subgraph connectivity verification (\Cref{sec:lower-bound}), a problem which readily reduces to many other distributed optimization problems~\cite{dassarma2012distributed}.
This universal lower bound holds even when the topology is known, i.e., even in the supported CONGEST model.

\smallskip

The parameter $\shortcut{G}$ is particularly notable since it is a key parameter in the running time of distributed algorithms that are based on the so-called ``low-congestion shortcut framework'' (see \Cref{sec:shortcuts}). A long line of work has shown that many distributed optimization problems can be solved in $\tilde{O}(\shortcut{G})$ rounds if one can efficiently construct near-optimal shortcuts. However, such constructions were only known for special types of graphs. We obtain efficient constructions for \emph{all} graphs in the known-topology settings by connecting the question to recent advancements in hop-constrained oblivious routings~\cite{ghaffari2020hop} (see \Cref{sec:shortcut-construction}). Putting this together, we obtain the following result.

\begin{wrapper}
\begin{restatable}{theorem}{thmUniversallyOptimal}[Informal]\label{thm:universally-Optimal}
  The problems minimum-spanning tree, $(1+\eps)$-minimum cut, sub-graph connectivity, various approximate shortest path problems (and more) admit a universally-optimal supported-CONGEST algorithm based on the low-congestion shortcut framework.
\end{restatable}
\end{wrapper}

Moreover, we identify several different graph parameters beside $\shortcut{G}$ that are universal CONGEST lower bounds (for the same set of problems). While the pararmeters are equivalent up to $\tilde{O}(1)$ factors, they provide different interpretations of the barriers that preclude fast algorithms for distributed network optimization.

\subsection{Related Work}

Probably the most well-studied global optimization problem in the distributed message-passing literature, and the one that best illustrates the search for universal optimality, is the minimum spanning tree (MST) problem. It was precisely the study of this problem which initiated the quest for universal optimality, as put forth in 1993 by Garay, Kutten, and Peleg~\cite{garay1993sublinear, garay1998sublinear}.
This problem was first studied in a distributed setting in the seminal work of \citet*{gallager1983distributed}, who gave an $O(n\log n)$-round MST algorithm. This was later improved  by \citet{awerbuch1987optimal} to $O(n)$ rounds, which is existentially optimal in $n$.
\citet{garay1998sublinear}, advocating for a more refined analysis, moved closer to the universal lower bound of $\Omega(D)$, giving an $\tilde{O}(D + n^{0.613})$-round MST algorithm. This was improved to $\tilde{O}(D + \sqrt{n})$ by \citet{kutten1998fast}. \citet{peleg2000near} constructed networks proving this bound is also existentially optimal in $n$ and $D$.
These networks were then used to prove lower bounds for approximate MST by \citet{elkin2006unconditional}, and many other problems by \citet{dassarma2012distributed}. 
Many algorithms matching this existentially-optimal $\tO(D+\sqrt{n})$ upper bound were obtained over the years \cite{kutten1998fast,elkin2006faster,pandurangan2017time,elkin2017simple}, including using the low-congestion shortcut framework~\cite{ghaffari2016distributed}.

\smallskip
The above results foreshadowed much work on studying other graph parameters which allow for improved running time for the MST problem. 
One example is restricting the diameter. 
For example, for graphs of diameter 1 (i.e., the congested clique model), a sequence of works \cite{lotker2003mst,hegeman2015toward,ghaffari2016mst} culminated in an $O(1)$-time algorithm \cite{jurdzinski2018mst,nowicki2021deterministic}.
For small-constant diameter, \citet{lotker2006distributed}~gave an $O(\log n)$ algorithm for diameter-2 graphs, and $\tilde{\Omega}(\sqrt[3]{n})$ and $\tilde{\Omega}(\sqrt[4]{n})$ lower bounds for graphs of diameter 3 and 4, with algorithms matching these bounds recently obtained using the low-congestion shortcut framework \cite{kitamura2019low}. 
Indeed, the shortcut framework has been the driving force behind numerous improved results for restricted graph families \cite{ghaffari2016distributed,haeupler2016low,haeupler2016near,haeupler2018minor,ghaffari2018new,ghaffari2017distributed,kitamura2019low,haeupler2018round}. 
For most of these results, the worst-case shortcut quality of a graph in the graph family serves as an upper bound for these algorithms' running time. 
Our work shows that shortcut quality is precisely the optimal running time for \emph{any} graph, 
proving that this graph parameter is a universal lower bound for distributed algorithms, and that this lower bound is achievable algorithmically by efficient supported-CONGEST algorithms.

\smallskip

\noindent\textbf{The power of preprocessing.} In this work we study CONGEST algorithms~\cite{peleg2000distributed}, both under the assumption that the topology is unknown or known to the nodes.
The latter is the \emph{supported CONGEST} model, introduced by 
\citet{schmid2013exploiting}, who were motivated by Software-Defined Networking (SDN). 
As they argued, in SDN enabled networks, the underlying communication topology is known, while the input (e.g., edge weights, subgraph to test connectivity of, etc) may vary. 
It is therefore natural to preprocess the graph in advance in order to \emph{support} the solution of possible inputs defined on this graph.
In this model \citet{ghaffari2015distributed} gave a polynomial-time preprocessing $k$-broadcast algorithm which is optimal among \emph{routing-based} distributed algorithms on the given input.
It was noted in \cite{schmid2013exploiting,foerster2019does,foerster2019power} that, while preprocessing intuitively seems very powerful, CONGEST lower bounds~\cite{dassarma2012distributed,abboud2019smaller} generally hold even in the supported model. Due to this, \cite{foerster2019does} asks whether preprocessing offers any benefit at all, and they offer several (specifically-made) tasks that do exhibit a separation. However, the question remains open for well-studied problems in the field. This paper offers a partial answer to this question that, up to efficient construction of near-optimal shortcuts, preprocessing does not help for many well-known problems.

\medskip

\noindent\textbf{Strengthened notions of optimality.}
Various notions of optimality similar to instance optimality have been studied in depth in many fields, and some algorithms achieving these desired properties are known in various computational models. Instance optimality with respect to a certain class of algorithms has been proven in aggregation algorithms for database systems \cite{fagin2003optimal}, shared memory distributed algorithms \cite{dwork1986knowledge}, 
geometric algorithms \cite{afshani2017instance}, and distribution testing and learning algorithms \cite{valiant2017automatic,valiant2016instance}. Indeed, the entire field of online algorithms concerns itself with the notion of competitive analysis, which can be seen as a form of instance optimality~\cite{manasse1990competitive,borodin2005online}.
For other models, stronger notions of optimality were long sought after, but remain elusive. For example, one of the oldest open questions in computer science is the dynamic optimality conjecture of \citet{sleator1985self}, which states that splay trees are instance-optimal among all binary search trees.
The problems studied here join the growing list of problems in various computational models for which such instance-optimal algorithms are known. This notion of instance optimality was discussed in the distributed setting by Elkin~\cite{elkin2006faster}, who achieved instance optimality with respect to coarsening will-maintaining protocols in the LOCAL model (see \Cref{sec:optimality-formal-definitions} for a comparison between this work and our results).

\subsection{Paper Outline}
We define the computational models and problems studied in this paper, and discuss necessary technical background in \Cref{sec:prelims}.
We then formalize the notion of universal optimality in \Cref{sec:universal-optimality}, and point out some delicate points concerning this notion.
In \Cref{sec:technical-overview} we give a technical overview of the paper, presenting the key ideas behind our matching bounds for supported CONGEST, and point out a number of asymptotically equivalent tight universal bounds which follow from our work.
We then substantiate our universal lower bound in \Cref{sec:lower-bound}, deferring the most technically involved part of this lower bound---proving the existence of disjointness gadgets---to \Cref{sec:constructing-disjointness}. We then present our matching supported CONGEST upper bound in terms of shortcut quality in \Cref{sec:shortcut-construction}. We conclude with a discussion of open questions in \Cref{sec:conclusion}. For the sake of legibility of the main body, a number of proofs are deferred to appendices, most notably a discussion of sub-diameter bounds in the supported CONGEST model, which is deferred to \Cref{sec:sub-diam}. (See also discussion in \Cref{sec:optimality-formal-definitions}.)


 
\section{Preliminaries}\label{sec:prelims}

\subsection{Models and Problems}\label{sec:models-and-problems}

Our discussion will focus on the following two models of communication.

\smallskip

\noindent \textbf{CONGEST~\cite{peleg2000distributed}.} In this setting, a network is given as a connected undirected graph $G = (V, E)$ with diameter $D$ and $n := |V|$ nodes. Initially, nodes know $n$ and $D$ and their unique $O(\log n)$-bit ID, IDs of their neighbors, and their problem-specific inputs. Communication occurs in synchronous rounds; during a round, each node can send $O(\log n)$ bits to each of its neighbors. 
The goal is to design protocols that minimize the number of rounds until all nodes are guaranteed to output a solution (some global function of the problem-specific inputs) and terminate.

\smallskip

\noindent\textbf{Supported CONGEST~\cite{schmid2013exploiting}.} This setting is same as the classic CONGEST setting, with the addition that each node knows all unique IDs and the entire topology $G$ at the start of the computation. Note that this is equivalent to the CONGEST where $\poly(n)$-round preprocessing is allowed before the problem-specific input is revealed.

\smallskip

The following problems will be used throughout the paper. Notably, the spanning connected subgraph verification problem will serve as the core of our lower bounds.

\smallskip

\noindent\textbf{Spanning connected subgraph~\cite{dassarma2012distributed}.} A subgraph $H$ of $G$ is specified by having each node known which of its incident edges belong to $H$. The problem is solved when all nodes know whether or not $H$ is connected and spans all nodes of $G$.

\smallskip

\noindent\textbf{MST, shortest path, min-cut.} Every node knows the $O(\log n)$-bit weights of each of its edges. The problem is solved when all nodes know the value of the final solution (i.e., the weight of the MST/min-cut/shortest path), and which of its adjacent edges belong to the solution.

\smallskip

These problem definitions follow~\cite{dassarma2012distributed} in requiring that all nodes know the weight of the final solution. This is needed in the reductions from the spanning connected subgraph problem given in \cite{dassarma2012distributed}, allowing us to leverage them in this paper. Our results seamlessly carry over to the arguably more natural problem variant where each node only needs to know which of its incident edges are in the solution (see \Cref{sec:sub-diam} for details).

\subsection{The Low-Congestion Shortcut Framework}\label{sec:shortcuts}

In this section we briefly summarize the \emph{low-congestion shortcut framework}~\cite{ghaffari2016distributed}. The framework is the state-of-the-art for message-passing algorithms for all global network optimization problems we study that go beyond worst-case topologies. This framework was devised to demonstrate that the $\tilde{\Omega}(\sqrt{n})$ lower bound~\cite{dassarma2012distributed,peleg2000near} does not hold for at least some natural topologies such as planar graphs. Since then, shortcuts have been used to achieve $o(\sqrt{n})$ running times on many different graph topologies (see \cite{haeupler2016low} for an overview).

\smallskip

The framework introduces the following simple and natural communication problem, called \emph{part-wise aggregation}.

\begin{restatable}{definition}{partwiseAggr}\label{def:partwiseAggr}[Part-wise aggregation]
  Given disjoint and connected subsets of nodes $P = (P_1, \ldots, P_k)$, where $P_i \subseteq V$ are called parts, and (private) $O(\log n)$-bit input values at each node, compute within each part in parallel a simple aggregate function. For instance, each node may want to compute the minimum value in its part.
\end{restatable}

The part-wise aggregation problem naturally arises in divide-and-conquer algorithms, such as Bor\r{u}vka's MST algorithm~\cite{Boruvka26}, in which a network is sub-divided and simple distributed computations need to be performed in each part. More importantly and surprisingly, many other, seemingly unrelated, distributed network optimization problems like finding approximate min-cuts~\cite{ghaffari2016distributed}, a DFS-tree~\cite{ghaffari2017near}, or solving various (approximate) shortest path problems~\cite{haeupler2018faster} similarly reduce to solving $\tilde{O}(1)$ part-wise aggregation instances. 

\smallskip

Unfortunately, the parts in such applications can be very ``long and windy'', inducing subgraphs with very large strong diameter, even if the underlying network topology has nice properties and a small diameter. Therefore, in order to communicate efficiently, parts have to utilize edges from the rest of the graph to decrease the number of communication hops (i.e., the \emph{dilation}). On the other hand, overusing an edge might cause \emph{congestion} issues. Balancing between congestion and dilation naturally leads to the following definition of \emph{low-congestion shortcuts}~\cite{ghaffari2016distributed}. 
\begin{restatable}{definition}{shortcutFormal}\label{def:shortcutFormal}[Shortcut quality]
  A shortcut for parts $(P_1, \ldots, P_k)$ is $( H_1, \ldots, H_k )$, where $H_i$ is a subset of edges of $G$. The shortcut has \textbf{dilation} $d$ and \textbf{congestion} $c$ if (1) the diameter of each $G[P_i] \cup H_i$ is at most $d$ (i.e., between every $u, v\in P_i$ there exists a path of length at most $d$ using edges of $G[P_i] \cup H_i$), and (2) each edge $e$ is in at most $c$ different sets $H_i$. The \textbf{quality} of the shortcut is $Q = c+d$.
\end{restatable}

Classic routing results by Leighton, Maggs and Rao~\cite{leighton1994packet} show that such a shortcut allows for the part-wise aggregation to be solved in $\tilde{O}(c+d)$ rounds, even distributedly. This motivates the definition of quality. We say a network topology $G$ admits low-congestion shortcuts of quality $Q$ if a shortcut with quality $Q$ exists for \emph{every} partition into disjoint connected parts of $G$'s nodes. Denoting the minimum such quality $Q$ by $\shortcut{G}$, this approach leads to algorithms whose running times are parameterized by $\shortcut{G}$.

\begin{restatable}{lemma}{thmShortcutsImplyMSTetal}\label{thm:shorcutsimplyMSTetal}
  [\cite{ghaffari2016distributed,dassarma2012distributed,haeupler2018faster}]
  If $G$ admits $Q$-quality shortcuts for every (valid) set of parts and such shortcuts are computable by a $T$-round CONGEST algorithm, then $G$ has $\tilde{O}(Q + T)$ round CONGEST algorithms for  minimum spanning tree, $(1+\eps)$-min-cut, approximate shortest-paths, and various other problems.
\end{restatable}

Several natural classes of network topologies, including planar networks~\cite{ghaffari2016distributed}, networks with bounded genus, pathwidth, treewidth~\cite{haeupler2016near,haeupler2016low,haeupler2018round} or expansion~\cite{ghaffari2017distributed}, and all minor-closed network families~\cite{haeupler2018minor} admit good shortcuts which can be efficiently constructed (\cite{haeupler2016low,haeupler2018round,ghaffari2017distributed}). Along with \Cref{thm:shorcutsimplyMSTetal}, this implies ultra-fast message-passing algorithms, often with $\tilde{O}(1)$ or $\tilde{O}(D)$ running times, for a wide variety of network topologies and network optimization problems. In this work we show that low-congestion shortcut-based algorithms are optimal on every network---yielding universally optimal algorithms for the problems studied here (and many more).




\subsection{Moving Cuts}

In this section we describe \emph{moving cuts}, a useful tool for proving distributed information-theoretic lower bounds. Moving cuts are used to lift strong unconditional lower bounds from the classic communication complexity setting into the distributed setting. This approach was used to prove existentially-optimal (in $n$ and $D$) lower bounds in \citet{dassarma2012distributed}, and moving cuts can be seen as a generalization of their techniques. Moving cuts were only explicitly defined in \cite{haeupler2020network}, where they were used to prove network coding gap for simple pairwise communication tasks. (More on such tasks in \Cref{sec:comm-moving-cuts}.)

\smallskip

Before defining moving cuts, we briefly discuss the communication complexity model and distributed function computation problems.

\smallskip

\textbf{Distributed computation of a Boolean function $f$}. 
In this problem, two distinguished (multi-)sets of nodes $\{s_i \in V\}_{i=1}^k$ and $\{t_i \in V\}_{i=1}^k$ are given. For a given function $f : \{0,1\}^k \times \{0,1\}^k \to \{0, 1\}$ and inputs $x, y \in \{0,1\}^k$, we want every node in $G$ to learn $f(x, y)$. However, $x_i$ and $y_i$ are given only to $s_i$ and $t_i$ as their respective private inputs.
Nodes in $G$ have access to shared random coins. 
We are interested in the worst-case running time to complete the above task in the supported CONGEST model. 
The problem is motivated by the (classic) \emph{communication complexity model}, which is the special case of CONGEST with a two-node, single-edge graph, where Alice controls one node (with $k$ input bits) and Bob controls the other node (ditto), and \emph{single-bit} messages are sent in each round.  The time to compute $f$ in this model is referred to as its communication complexity.
In this paper we are mostly interested in the $k$-bit \emph{disjointness function}, $\disj:\{0,1\}^k\times \{0,1\}^k\rightarrow \{0,1\}$, given by $\disj(x, y) = 1$ if for each $i\in [k]$, we have $x_i \cdot y_i = 0$, and $\disj(x, y) = 0$ otherwise. I.e., if $x$ and $y$ are indicator vectors of sets, this function indicates whether these sets are disjoint.
This function is known to have communication complexity $\Theta(k)$ \cite{razborov1990distributional, dassarma2012distributed}.

\smallskip

Having defined the distributed function computation model, we are now ready to define our lower-bound certificate on the time to compute a function between nodes $\{s_i\}_{i=1}^k$ and $\{t_i\}_{i=1}^k$:~a~\emph{moving~cut}.

\begin{definition}[\cite{haeupler2020network}]
  \label{def:moving-cut}
  Let $S = \{(s_i, t_i)\}_{i=1}^k$ be a set of source-sink pairs in a graph $G = (V, E)$.  A \textbf{moving cut} for $S$ is an assignment of positive integer edge lengths $\ell : E \to \mathbb{Z}_{\ge 1}$. We say that: 
  \begin{enumerate}[(i)]
  \item $\ell$ has \textbf{capacity} $\gamma := \sum_{e \in E} (\ell_e - 1)$;
  \item $\ell$ has \textbf{distance} $\beta$ when $\dist_\ell(\{s_i\}_{i=1}^k, \{t_j\}_{j=1}^k) \geq \beta$, i.e., the $\ell$-distance between all sinks and sources is at least $\beta$.
  \end{enumerate}
\end{definition}

The utility of moving cuts is showcased by the following lemma.

\begin{restatable}{lemma}{simulation}\label{lemma:moving-cut-simulation}  
  If $G$ contains a moving cut for $k$ pairs $S=\{(s_i,t_i)\}_{i=1}^k$ with distance at least $\beta$ and capacity strictly less than $k$, then distributed computation of $\disj$ between $\{s_i\}_{i \in [k]}$ and $\{t_i\}_{i \in [k]}$ takes $\tOmega(\beta)$ time. 
  This lower bound holds even for bounded-error randomized algorithms that know $G$ and $S$.
\end{restatable}

Broadly, \Cref{lemma:moving-cut-simulation} follows from a simulation argument. Given a sufficiently-fast $\tO(\beta)$-time distributed algorithm for $\disj$ between $\{s_i\}_{i=1}^k$ and $\{t_i\}_{i=1}^k$ and a moving cut for $\{(s_i,t_i)\}_{i=1}^k$ of capacity less than $k$ and distance $\beta$, we show how to obtain a communication complexity protocol with a sufficiently small complexity $O(k)$ to contradict the classic $\Omega(k)$ communication complexity lower bound for disjointness. This yields the lower bound.
The full proof follows the arguments (implicitly) contained in \cite{dassarma2012distributed,haeupler2020network}, and is given in \Cref{sec:prelims-appendix} for completeness. 

\smallskip 
\Cref{lemma:moving-cut-simulation} motivates the search for moving cuts of large distance and bounded capacity. For a fixed set of $k$ pairs $S$, we define $\movingcut{S}$ to be the largest distance $\beta$ of a moving cut for $S$ of capacity strictly less than $k$.

\subsection{Relation of Moving Cuts to Communication}\label{sec:comm-moving-cuts}

Consider the simple communication problem for pairs $S=\{(s_i,t_i)\}_{i=1}^k$, termed \emph{multiple unicasts}. In this problem, each $s_i$ has a single-bit message $x_i$ it wishes to transmit to $t_i$. We denote by $\comm{S}$ the time of the fastest algorithm for this problem which knows $G$ and $S$ (but not the messages). One natural way to solve this problem is to store-and-forward (or ``route'') the messages $x_i$ through the network. We denote the fastest such algorithm's running time by \routing{S}. While faster solutions can be obtained by encoding and decoding messages in intermediary nodes, prior work has shown the gap between the fastest routing and unrestricted (e.g., coding-based) algorithms is at most $\tO(1)$~\citep{haeupler2020network}.
Indeed, the following lemma asserts as much, and shows that moving cuts characterize the time required to complete multiple unicasts. As the model and terminology of \cite{haeupler2020network} is slightly different from ours, we provide a proof of this lemma in \Cref{sec:prelims-appendix}.

\begin{restatable}{lemma}{pairwise}(\cite{haeupler2020network})\label{moving-cuts-communication}
  For any set of pair $S$, there exists a subset of pairs $S' \subseteq S$ such that
  $$\movingcut{S'} = \tilde{\Theta}(\comm{S}) = \tilde{\Theta}(\routing{S}).$$
\end{restatable}

Furthermore, routing algorithms (and, by extension, moving cuts) are intricately related to shortcuts. We first extend shortcuts to (not necessarily connected) pairs in the straightforward way: given a set of pairs $S = \{(s_i, t_i)\}_{i=1}^k$ we say that a set of paths $\{H_i\}_{i=1}^k$ with endpoints $\{s_i, t_i\}$ are $q$-quality shortcuts if both their dilation (length of the longest path) and congestion (maximum number of paths containing any given edge) are at most $q$. We define $\pairshortcut{S}$ as the minimum shortcut quality achievable for $S$. It is a relatively straightforward fact that $\routing{S}$ and $\pairshortcut{S}$ are essentialyl equal (up to polylogarithmic factors).

\begin{lemma}\label{LMR}(\cite{leighton1994packet})
  For any set of pairs $S$, we have that $\routing{S} =  \tilde{\Theta}(\pairshortcut{S})$.
\end{lemma}
\begin{proof}
  Given $k$ pairs $S$ and a set of shortcut paths $\{ H_i \}_{i=1}^k$ for $S$ of quality $q$, one can use the famous random delay routing~\cite{leighton1994packet}: each message is delayed a uniformly random number between $0$ and $q$, then all edges in each round have congestion at most $\tilde{O}(1)$ w.h.p., hence by subdividing each round into $\tilde{O}(1)$ mini-round it is guaranteed that each message advances one hop per round. Such a routing algorithm completes in $\tilde{O}(q)$ time.

  On the other hand, given a routing algorithm of completion time $q$, one can ``trace'' each packet corresponding to the pairs $(s_i, t_i)$ to have this taken path be the shortcut $H_i$. It is immediate that the congestion and dilation of $\{ H_i \}_{i=1}^k$ is at most $q$.
\end{proof}

While the above statements hold for all sets of pairs, we will mostly be concerned with sets of pairs $S$ which can be connected by vertex-disjoint paths in $G$---which we refer to as \emph{connectable} sets of pairs. We argue that worst-case connectable pair sets characterize distributed optimization, and hence we define $\movingcut{G}:=\max\{\movingcut{S} \mid S \text{ is connectable}\}$, and analogously for 
$\comm{G}$, $\routing{G}$ and $\pairshortcut{G}$.
These definitions together with \cref{moving-cuts-communication,LMR} immediately imply the following relationships, proven in \Cref{sec:prelims-appendix} for completeness.
\begin{restatable}{lemma}{relationships}\label{relationships}
	For any graph $G$, we have 
$$\movingcut{G} = \tilde{\Theta}(\comm{G}) = \tilde{\Theta}(\routing{G}) = \tilde{\Theta}(\pairshortcut{G}).$$
\end{restatable}

\subsection{Oblivious Routing Schemes}

Here we revisit the multiple unicasts problem for a set of pairs $S=\{ (s_i, t_i) \}_{i=1}^k$, whose complexity is captured by the parameter $\comm{S}$. 
By lemmas \ref{moving-cuts-communication} and \ref{LMR}, this parameter is also equal (up to polylog terms) to the best shortcut quality for these pairs, $\pairshortcut{S}$: that is, the minimum congestion+dilation over all sets of paths connecting each $(s_i,t_i)$ pair.
If all pairs are aware of each other and the topology, multiple unicasts is therefore solvable optimally (up to polylogs), using standard machinery (see \cite{haeupler2020network}).
However, what if each $s_i$ needs to transmit its message to $t_i$ without knowing other pairs $(s_j, t_j)$ in the network? 
Can we achieve such near-optimal routing with the choice of $s_i$ being \emph{oblivious} to the other pairs?
The following definitions set the groundwork needed to describe precisely such an oblivious routing scheme.

\begin{definition}
	A \textbf{routing scheme} for a graph $G = (V, E)$ is a collection $R = \{ R_{s, t} \}_{s, t \in V}$ where $R_{s, t}$ is a distribution over paths between $s$ and $t$. The routing scheme $R$ has \textbf{dilation} $d$ if for all $s, t \in V$, all paths $p$ in the support of $R_{s, t}$ have at most $d$ hops.
\end{definition}

Fix demands between pairs $\calD : V \times V \to \{0, 1, \ldots, n^{O(1)}\}$. We say a routing scheme $R = \{ R_{s, t} \}_{s, t}$ has a \emph{fractional routing congestion} w.r.t.~$\calD$ of $\cong(\calD, R) := \max_{e \in E} \E_{p \sim R_{s,t}} \left[ \calD_{s, t} \cdot \mathds{1}[e \in p ] \right ]$. 
We denote by $\opt\at{h}(\calD)$ the minimum $\cong(\calD, R)$ over all routing schemes $R$ with dilation $h$.
The (hop-constrained) oblivious routing problem asks whether there exists a routing scheme $R$ which is oblivious to the demand (i.e., does not depend on $\calD$), but which is nevertheless competitive with the optimal (hop constrained) fractional routing congestion over all demands. 
\begin{definition}
	An \textbf{$h$-hop oblivious routing scheme} for a graph $G = (V, E)$ with \textbf{hop stretch} $\beta \ge 1$ and \textbf{congestion approximation} $\alpha \ge 1$ is a routing scheme $R$ with dilation $\beta \cdot h$ and which for all demands $\calD : V \times V \to \mathbb{R}_{\ge 0}$ satisfies
	\begin{align*}
	\cong(\calD, R) \le \alpha \cdot \opt\at{h}(\calD) .
	\end{align*}
\end{definition}

The following lemma, which follows by standard Chernoff bounds (see \Cref{sec:prelims-appendix}) motivates the interest in such an oblivious routing in the context of computing optimal shortcuts for a set of pairs.
\begin{restatable}{lemma}{shortcutsrouting}\label{shortcuts-from-routing}
	Let $G$ be a graph, and $S=\{(s_i,t_i)\}$ be a set of pairs. 
	Let $h\in [Q,2Q)$, where $Q=\pairshortcut{G}$.
	Finally, let $R$ be an $h$-hop oblivious routing for $G$ with hop stretch $\beta \ge 1$ and congestion approximation $\alpha\geq 1$. Then, sampling a path $p \sim R_{s_i,t_i}$ for each $(s_i,t_i)\in S$ yields $O(\log n)\cdot \max\{\alpha,\beta\}\cdot Q$ shortcuts for $S$ w.h.p.
\end{restatable}

The seminal result of \citet{racke2008optimal} asserts that for $h=n$ (i.e., no hop constraint), a $h$-hop oblivious routing scheme with congestion approximation $\alpha=O(\log n)$ exists.
The main result of \cite{ghaffari2020hop} asserts that good hop-constrained oblivious routing exists for \emph{every} hop bound $h$.
\begin{restatable}{theorem}{thmGeneralRouting}\label{thmGeneralRouting}(\cite{ghaffari2020hop})
	For every graph $G = (V, E)$ and every $h \ge 1$, an $h$-hop oblivious routing with hop stretch $O(\log^6 n)$ and congestion approximation $O(\log^2 n \cdot \log^2 (h \log n))$ is computable in polytime.
\end{restatable}


\section{Universal Optimality}\label{sec:universal-optimality}

In this section we discuss different notions of optimality either explicitly or implicitly apparent  in the literature, starting with a high-level description. We then provide a formal definition of universal and instance optimality.

\subsection{Different Notions of Optimality}

Existential optimality (discussed below) is the standard worst-case asymptotic notion optimality with respect to simple graph parameters prevalent in both distributed computing and throughout theoretical computer science at large. 

\medskip

\noindent\textbf{Existential Optimality.} The MST algorithm of Kutten and Peleg~\cite{kutten1998fast}, which terminates in $\tilde{O}(D + \sqrt{n})$ rounds on every network with $n$ nodes and diameter $D$, is optimal with respect to $n$ and $D$ in the following \emph{existential sense}. For every $n$ and $D = \Omega(\log n)$ there \emph{exists} a network $G$ with $n$ nodes and diameter $D$, and a set of inputs on $G$, such that any algorithm that is correct on all inputs requires at least $\tilde{\Omega}(D + \sqrt{n})$ rounds.\footnote{The bounds slightly change when $D = o(\log n)$, e.g., see \cite{kitamura2019low, chuzhoy2020packing,lotker2006distributed}.} However, the drawback of existential optimality is immediate: it says nothing about the performance of an algorithm compared to what is achievable on networks of interest. For example, in every planar graph one can compute the MST in $\tilde{O}(D)$ rounds, outperforming the $\tilde{\Omega}(D + \sqrt{n})$-bound when $D \ll \sqrt{n}$.
Another, less immediate, drawback is that existential optimality crucially depends on the parameterization. If we parameterize only by $n$, one would not need to look past the $O(n)$ MST algorithm described by Awerbuch \cite{awerbuch1987optimal}. In the other extreme, one could start searching for existential optimality with respect to an ever-more-complicated set of parameters, ad nauseam.

Due to these drawbacks, Garay, Kutten, and Peleg~\cite{garay1998sublinear} informally proposed to study stronger notions of optimality (see their in \Cref{sec:intro-better-optimality}).
Based on the quote from \cite{garay1998sublinear}, we define the following two notions of optimality.

\medskip

\noindent\textbf{Instance Optimality.}
We say that an algorithm is \emph{instance optimal} if it is $\tilde{O}(1)$-competitive with every other always-correct algorithm on \emph{every network topology} $G$ and \emph{every valid input}.

\smallskip

\noindent\textbf{Universal Optimality.} We define universal optimality as a useful middle ground between (weaker) existential and (often unachievable) instance optimality. We say that an algorithm is \emph{universally optimal} if for \emph{every network} $G$, the \emph{worst-case running time across all inputs} on $G$ is $\tilde{O}(1)$-competitive with the worst-case running time of any other always-correct algorithm running on $G$. In other words, we are universal over the network topology $G$, but worst-case over all inputs on $G$.

An immediate benefit of universal optimality is that its definition is independent of any parameterization. Moreover, if a universally-optimal algorithm for a problem $\Pi$ terminates in times $X_{\Pi}(G)$ for graph $G$, this implies that $X_{\Pi}(\cdot)$ is the fundamental graph parameter that inherently characterizes the hardest barrier in $G$ that prevents faster algorithms from being achievable. The term \emph{universal-network optimality} would be somewhat more descriptive given that we are optimal for every network. However, throughout this paper we chose to keep the term ``universal optimality'' in homage to the original paper~\cite{garay1993sublinear} which states the problem of identifying ``inherent graph parameters'' associated with different problems, and referred to the algorithms matching those bounds as universally optimal.

It is not clear whether instance universal optimality can be achieved. Indeed, it could be possible that each instance or network have a single tailor-made algorithm which solves it in record time, at the cost of being much slower on other instances or networks. Requiring a single uniform algorithm to compete with each of these fine-tuned algorithms on every instance or network simultaneously might simply not be possible. Indeed, we show that this barrier prevents instance-optimal MST algorithms from existing in (supported) CONGEST model. Surprisingly, we show that universal optimality for MST and many other problems is achievable in supported CONGEST.

\subsection{Formal Definitions}\label{sec:optimality-formal-definitions}

In this section we give formal definitions of instance and universal optimality. To our knowledge, this is the first paper that clearly separates these two notions.

For some problem $\Pi$, we say an algorithm $\calA$ is \emph{always correct} if it terminates with an correct answer on every instance (i.e., every network $G$ and every input $I$), and let $\altA_{\Pi}$ be the class of always-correct algorithms for $\Pi$. Denote by $T_{\calA}(G, I)$ the running time of algorithm $\calA \in \altA_{\Pi}$ on graph $G$ and problem-specific input $I$, and let $\max_{I} T_{\calA}(G, I)$ be the worst-case running time of algorithm $\calA \in \altA_{\Pi}$ over all inputs supported on graph $G$. We define universal optimality and instance optimality for algorithms in model $\calM$ (e.g., CONGEST, supported CONGEST, LOCAL, etc...) as follows:

\begin{wrapper}
  \begin{restatable}{definition}{defuniopt}(Instance Optimality)
    An always-correct model-$\calM$ algorithm $\calA$ for problem $\Pi$ is \textbf{instance optimal} if $\calA$ is $\tilde{O}(1)$-competitive with every always-correct model-$\calM$ algorithm $\calA'$ for $\Pi$ on every graph $G$, and every input $I$, i.e.,
    \begin{align*}
      \forall \calA' \in \altA_{\Pi},\, \forall G,\, \forall I\qquad T_{\calA}(G, I) = \tilde{O}(1)\cdot T_{\calA'}(G, I) .
    \end{align*}
  \end{restatable}
\end{wrapper}

\begin{wrapper}
  \begin{restatable}{definition}{defuniopt}(Universal Optimality)
    An always-correct model-$\calM$ algorithm $\calA$ for problem $\Pi$ is \textbf{universally optimal} if the worst-case running time of $\calA$ on $G$ is $\tilde{O}(1)$-competitive with that of every always-correct model-$\calM$ algorithm $\calA'$ for $\Pi$, i.e.,
    \begin{align*}
      \forall \calA' \in \altA_{\Pi},\, \forall G\qquad \max_I T_{\calA}(G, I) = \tilde{O}(1) \cdot \max_I T_{\calA'}(G, I) .
    \end{align*}
  \end{restatable}
\end{wrapper}

While similar on a surface level, these two notions have vastly different properties. Notably, in CONGEST, any instance optimal algorithm would need to compute an answer in $\tilde{O}(D)$ rounds.
\begin{lemma}\label{lemma:instance-optimality-is-always-D}
  For every problem $\Pi$ in CONGEST or supported CONGEST, every instance-optimal algorithm $\calA$ must terminate in $\tilde{O}(D)$ rounds for every instance on a graph of diameter $D$.
\end{lemma}
\begin{proof}
 For any instance, i.e., graph $G$ and input $I$, consider the algorithm $\calA'_{(G,I)}$ that simply has each node check a node with its ID is present in $G$ whether its own input and local neighborhood is identical

checks for a specific graph/input pair $(G',I')$ whether the graph/input $(G,I)$ are equal to $(G',I')$.
  In $O(D)$ rounds, any node whose neighborhood or part of the input are inconsistent with $I'$, notifies all nodes, by flooding a message through the network. After these $O(D)$ rounds, the solution for $(G,I)$ is output if $(G,I)=(G',I')$. Otherwise, a polynomial-time algorithm is run to aggregate $I$ and solve $(G,I)$. (This last step guarantees $\calA' \in \altA_{\Pi}$, i.e., is always correct.)
  That is, for any input $I'$ supported on a diameter $D$ graph $G'$, there exists an always-correct algorithm $\calA'$ for which $T_{\calA'}(G',I')=O(D)$. 
  Therefore, an instance-optimal algorithm $\calA$ must satisfy $T_{\calA}(G, I) = \tilde{O}(1) \cdot T_{\calA'}(G, I) = \tilde{O}(D)$.
\end{proof}

\Cref{lemma:instance-optimality-is-always-D} implies that instance optimality is unattainable in the supported CONGEST model for the problems we study: each MST algorithm in supported CONGEST requires $\tilde{\Omega}(\sqrt{n})$ rounds on some instance supported on a network of diameter $D = O(\log n)$~\cite{dassarma2012distributed}.
So, while the notion of instance optimality has merit for other problems or models, for the problems studied here in supported CONGEST, this notion is unachievable.
On the other hand, we show that universal optimality can be achieved for the problems we study in supported CONGEST.

To illustrate some differences between instance and universal optimality, we note that the latter does not suffer the same issues that plagued instance optimality: consider the worst-case network $G_{WC}$ which we know suffers a $\tilde{\Omega}(\sqrt{n})$ supported CONGEST lower bound (on some input)~\cite{dassarma2012distributed}, hence for any always-correct $\calA'$ we have that $\max_{I} T_{\calA'}(G_{WC}, I) \ge \tilde{\Omega}(\sqrt{n})$. When presented with $G_{WC}$, a universally-optimal algorithm $\calA$ must terminate in $\tilde{O}(\sqrt{n})$ rounds, which is achievable even with the existentially-optimal $\tilde{O}(D + \sqrt{n})$ algorithm~\cite{kutten1998fast}.

Another illustrative example: suppose that $\calM$ is CONGEST and some problem $\Pi$ allows for an $O(D)$ CONGEST algorithm when the underlying network is a tree. Then any universally-optimal algorithm for $\Pi$ must complete in $\tilde{O}(D)$ rounds when $G$ is a tree since an always-correct $\calA'$ could first check if the network $G$ is a tree in $O(D)$ rounds, and then proceed to run the optimal tree algorithm (if $G$ is not a tree $\calA'$ runs a slower $\poly(n)$ algorithm, ensuring its correctness).

Next, we show that the definition of universal optimality allows to formalize the often-stated remark that solving \emph{global} problems like the minimum spanning tree requires $\Omega(D)$ rounds in CONGEST. 
Note that the following $\Omega(D)$ universal lower bound crucially depends on requiring that always-correct algorithms $\calA \in \altA_{\Pi}$ have to be correct on \emph{every} topology, and not just the topology $G$ it is evaluated on.

\begin{lemma}\label{lemma:mst-congest-has-diameter-lower-bound}
  For every always-correct CONGEST MST algorithm and for graph $G$ of diameter $D$, the worst-case running time of $\calA$ on instances supported on $G$ is $\Omega(D)$, i.e., $\max_I T_{\calA}(G, I) \ge \Omega(D)$.
\end{lemma}
\begin{proof}
  If $n\leq 2$, or more generally $D=O(1)$, the observation is trivial. Suppose therefore that $n\geq 3$ and $D\geq 3$.
  Let $s$ and $t$ be two vertices of maximum hop distance in $G$. That is, $d_G(s,t)=D$.
  Consider a vertex $v$ at distance at least $\frac{D}{2}-1$ from both $s$ and $t$ (such a vertex must exist, else $d_G(s,t)\leq D-2$). 
  Fix a simple path $p:s\rightsquigarrow v \rightsquigarrow t$, and let $e=(u,v)$ be some edge in $p$ incident on $v$.
  We next consider two instances in two different graphs, $G$ and $G'$, obtained from $G$ by adding edge $(s,t)$. (Note that $(s,t)$ is not an edge in $G$, else $d_G(s,t)=1$.) The instances $I$ in $G$ which we consider assigns weights
  \begin{align*}
    w_{e'} = \begin{cases} 1 & e' \in p\setminus \{e\} \\
      2 & e'=e \\
      n & e' \notin p,
    \end{cases}
  \end{align*}
  while the instance $I'$ in $G'$ assigns the same weights to edges which also belong to $G$ and weight $w_{(s,t)}=1$ to $(s,t)$.
  By application of Kruskal's MST algorithm, it is easy to show that $e=(u,v)$ is in every MST of the instance $I$ in graph $G$, while it belongs to no MST of instance $I'$ in graph $G'$.
  However, after $\frac{D}{2}-2$ CONGEST rounds, no messages which are functions of the input of nodes $s$ and $t$ may reach $v$. Therefore, after $\frac{D}{2}-2$ rounds, node $v$ cannot distinguish whether the underlying topology is $G$ or $G'$, and it can therefore not distinguish between both instances.
  Consequently, it cannot determine whether $e=(u,v)$ is in an MST or not.
  Consequently, any always-correct MST CONGEST algorithm must spend at least $\frac{D}{2}-1$ rounds on any diameter $D$ graph.
\end{proof}

We give a short discussion of various aspects of our notions of optimality.

\medskip
\noindent\textbf{Universal optimality in supported CONGEST.} When $\calM$ is supported CONGEST, a universally-optimal algorithm $\calA$ can perform arbitrary computations on the network topology before the problem-specific input is revealed to it. This enlarges the space of possible universally-optimal algorithms compared to the (classic) CONGEST. On the other hand, the relative power of the ``competitor algorithm'' $\calA'$ is not significantly impacted between the two models; the argument behind the proof of \Cref{lemma:instance-optimality-is-always-D} implies that the running time of the best always-correct CONGEST and supported CONGEST algorithm on any input $I$ supported on $G$ never differ by more than an (often insignificant) $O(D)$ term.

\medskip
\noindent\textbf{Coarsening instance-optimal MST in LOCAL.} Elkin~\cite{elkin2006faster} defines the class of ``coarsening will-maintaining'' protocols as those that, in each round, maintain a set of edges which contain an MST, and eventually converge to the correct solution. The paper considers the LOCAL model (i.e., CONGEST with \emph{unlimited} message sizes) and concludes that one can construct instance-optimal (coarsening will-maintaining) algorithms. Specifically, the paper defines the so-called \emph{MST-radius} $\mu(G, I)$ (a function of both the network $G$ and the input $I$) and argues it is a lower bound for any always-correct LOCAL algorithm in the above class; the upper bound of $O(\mu(G,I))$ can also be achieved. 
A few notable differences between the definitions in Elkin's paper~\cite{elkin2006faster} and our results are imminent: in the former, protocols do not need to detect when to terminate, but rather converge towards the answer. This change of the model makes the results incomparable to ours---every always-correct CONGEST algorithm has a universal lower bound of $\Omega(D)$, which can often be significantly larger than the MST-radius. 

\medskip
\noindent\textbf{Is the diameter always a universal lower bound?} The results in this paper typically ignore additive $\tilde{O}(D)$ terms. For interesting models and problems for the scope of this paper, this choice can be formally justified with the universal lower bound of \Cref{lemma:mst-congest-has-diameter-lower-bound}. However, our universal and instance optimality formalism is interesting even in settings where sub-diameter results are possible, i.e., where the $\Omega(D)$ lower bound does not hold. For example, suppose that we define the MST problem to be solved when each node incident to an edge $e$ knows whether $e$ is part of the MST or not. In the known-topology setting (i.e., supported CONGEST), when $T$ is a tree, a universally-optimal algorithm takes $1 \ll D$ round, since each edge must be in the MST. This MST problem in supported CONGEST has the maximum diameter of a biconnected component, rather than the diameter of $G$, as a universal lower bound---see \Cref{sec:sub-diam} for a short exploration of such issues. The main body of this paper mostly defines such issues away by requiring all nodes to know the final output, making the problem harder than global aggregation, giving it a $\Omega(D)$ universal lower bound. We also note that this MST problem in (non-supported) CONGEST still has an $\Omega(D)$ universal lower bound.

\smallskip 

In the following section we outline our approach for proving the existence (and design) of universally optimal algorithms for the problems studied in this paper.


\section{Technical Overview}\label{sec:technical-overview}



In this section we outline the key steps towards obtaining universally-optimal algorithms in the supported CONGEST model, and highlight additional results implied by our work. 

The problem we use as our running example (and as the core of our lower bounds) is the \emph{spanning connected subgraph} verification problem (defined in \Cref{sec:prelims}).
By known reductions from spanning connected subgraph verification  \citet{dassarma2012distributed}, lower bound for the above problem extend to lower bounds for MST, cut, min-cut, s-source distance, shallow light trees, min-routing cost trees and many other problems as well as to any non-trivial approximations for these problems. 
In order to provide universal lower bounds for these problems, we therefore prove such universal lower bounds for this verification problem.

\subsection{Generalizing the existential lower bound to general topologies}

To achieve our results we first give a robust definition of a \emph{worst-case subnetwork}, which generalizes the pathological worst-case topology of the existential lower bound of \citet{dassarma2012distributed} to subnetworks in general graphs. This generalization builds on insights and crucial definitions from a recent work of the authors \cite{haeupler2020network}, which connects the CONGEST lower bound of~\citet{dassarma2012distributed} to network coding gaps for multiple unicasts. Once the new definition is in place it is easy to verify that the proof of \cite{dassarma2012distributed} generalizes to our worst-case subnetworks. Defining $\wcsub{G}$ to be the size of the largest such worst-case subnetwork in $G$ then gives a lower bounds for \emph{any} network, instead of just a single graph that is carefully chosen to facilitate the lower bound proof.\footnote{Indeed, \citet{dassarma2012distributed} state concerning their existential lower bound that ``The choice of graph $G$ is critical.''} One particularly nice aspect of this universal lower bound is that it brings the full strength and generality of the lower bound of \cite{dassarma2012distributed} to general topologies. In particular, it applies to a myriad of different optimization and verification problems, holds for deterministic and randomized algorithm alike, holds for known topologies, and extends in full strength to any non-trivial approximations. 
\begin{wrapper}
\begin{restatable}{lemma}{witnesslb}\label{lem:lower-bound-witness-intro}
  For any graph $G$, any always-correct supported CONGEST spanning connected subgraph verification algorithm $\calA_G$ takes $\tOmega(\wcsub{G} +  D)$ rounds on some input supported on $G$. This holds even if $\calA_G$ is randomized and knows $G$.
\end{restatable}
\end{wrapper}

\subsection{Shortcut quality is a universal lower bound}

A priori, it is not clear how strong or interesting the $\wcsub{G}$ lower bound is. By definition, it only applies to networks with subnetworks displaying similar characteristics to the pathological worst-case topology from \cite{dassarma2012distributed}, which seems very specific. Surprisingly, however, we prove an equivalence (up to polylog terms) between this graph parameter and several other graph parameters, including and most importantly a universal lower bound of $\shortcut{G}$. (We elaborate on these in \Cref{sec:characterizatoins}.) 

\begin{wrapper}
\begin{lemma}\label{T>=Q}
	For any graph $G$, $$\shortcut{G} = \tilde{\Theta}(\wcsub{G} + D).$$
\end{lemma}
\end{wrapper}

From this equivalence and our universal lower bound in terms of 
the worst-case subnetwork,
we obtain our main result: a universal lower bound in terms of the graph's shortcut quality.

\begin{wrapper}
	\begin{restatable}{theorem}{witnesslb}\label{thm:universally-lb}
		For any graph $G$, any always-correct message-passing supported CONGEST algorithm $\calA_G$ for spanning connected subgraph verification takes  $\tOmega(\shortcut{G})$ rounds on some input supported on $G$. This holds even if $\calA_G$ is randomized and knows $G$.
	\end{restatable}
\end{wrapper}

By standard reductions presented it \citet{dassarma2012distributed}, we deduce that $\shortcut{G}$ serves as a lower bound for various distributed optimization and verification problems.
  \begin{restatable}{corollary}{witnesslbcor}\label{thm:universally-lb-cr}
    Let $G$ be a graph, $\Pi$ be either MST, $(1+\eps)$-min-cut, or approximate shortest paths, and $\calA_G$ an always-correct supported-CONGEST algorithm for $\Pi$. Then, $\calA_G$ takes $\tOmega(\shortcut{G})$ rounds on at least one input supported on $G$.
  \end{restatable}

As a corollary of \Cref{thm:universally-lb} and the aforementioned reductions of \cite{dassarma2012distributed}, we find that the parameter $\shortcut{G}$ is also  a universal lower bound for the complexity of the very same optimization problems for which the low-congestion framework has already established algorithmic results with running times mostly depending on $\shortcut{G}$. 
Indeed, by \Cref{thm:shorcutsimplyMSTetal}, an algorithm constructing $\tilde{O}(1)$-approximately optimal shortcuts in time $\tilde{O}(\shortcut{G})$ would result in algorithms with running time 
$\tilde{O}(\shortcut{G})$, which would be universally optimal, by \Cref{thm:universally-lb}. 
We provide precisely such shortcut construction in the known topology setting.

\begin{wrapper}
\begin{restatable}{theorem}{shortcutsupported}\label{partwise-shortcut-supported} 
	There exists a supported CONGEST algorithm that, for any $k$ disjoint sets of connected parts $\{ P_i \subseteq V \}_{i=1}^k$ in a network $G$, constructs $\tilde{O}(\shortcut{G})$-quality shortcut on $\{ P_i \}_i$ in $\tilde{O}(\shortcut{G})$ rounds. 
\end{restatable}
\end{wrapper}

\subsection{Universal optimality in supported CONGEST}

The above results combined directly imply universally-optimal supported CONGEST algorithms for any problem that has a good shortcut-based distributed algorithm. 
\thmUniversallyOptimal*

\begin{proof}
  Fix a graph $G$, and let $Q:=\shortcut{G}$. By \Cref{partwise-shortcut-supported} , there exists a supported CONGEST algorithm which for any connected parts computes $\tilde{O}(Q)$-quality shortcuts in $\tilde{O}(Q)$ time.
  But then, by \Cref{thm:shorcutsimplyMSTetal}, there exists an algorithm for computing MST, $(1+\eps)$-min cut, approximate shortest paths, and spanning connected subgraph verification, all in $\tilde{O}(Q)$ rounds.
  Call the obtained algorithm $\calA$.
  That is,
  \begin{align*} 
    \max_I T_{\calA}(G,I) = \tilde{O}(\shortcut{G}).
  \end{align*}
  On the other hand, by \Cref{thm:universally-lb} and its \Cref{thm:universally-lb-cr}, for any algorithm $\calA'$, we have that 
  \begin{align*} 
    \max_I T_{\calA'}(G,I) = \tilde{\Omega}(\shortcut{G}).
  \end{align*}
  Combining the above, we find that $\max_I T_{\calA}(G,I) = \tO(1)\cdot  \max_I T_{\calA'}(G,I)$. 
  As the same holds for all graphs $G$, we conclude that $\calA$ is universally optimal.
\end{proof}


\subsection{Different characterizations of a topology's inherent distributed complexity}\label{sec:characterizatoins}

In this work we show that $\shortcut{G}$ is a tight universal lower bound for our problems. 
However, as mentioned before, identifying, understanding, and characterizing the aspects of a topology that influence and determine the complexity of distributed optimization problems is in itself a worthwhile goal. Indeed, there are a multitude of reasons why a detailed understanding of the relationship between topology and complexity is important. Among other reasons, it (a) can be important for the design of good networks, (b) might give important leads for understanding the structure of existing natural and artificial networks occurring in society, biology, and other areas, and (c) is necessary to provide quantitative and provable running time guarantees for universally-optimal algorithms run on a known topology $G$, beyond a simple ``it runs as fast as possible''.

\smallskip

As such, another important contribution of the tight lower bounds proven in this paper consists of giving different characterizations and ways to think about what makes a topology hard (or easy). 
For example, while $\wcsub{G}$ and $\shortcut{G}$ are both  quantitatively equal, the fact that they both characterize the complexity of distributed optimization lends itself to very different interpretations and conclusions.

\smallskip
Indeed, $\shortcut{G}$ can be seen as the best routing schedules for the partwise aggregation problem, which is the very natural communication primitive underlying distributed divide-and-conquer style algorithms (see, e.g., \cite{haeupler2016low}). Shortcut quality being a tight universal lower bound further demonstrates the key role partwise aggregation plays for distributed optimization algorithms, even to the extent that the complexity of many very different optimization tasks is dominated by how fast this simple aggregation procedure can be performed on a given topology. 

\smallskip
The tightness of $\wcsub{G}$ as a lower bound, on the other hand, points to the pathological network structure identified by Peleg and Rubinovich~\cite{peleg2000near,dassarma2012distributed} as indeed the only way in which a topology can be hard for optimization. Put otherwise, a topology is exactly as hard as the worst obstruction of this type within a network.  

\smallskip
As part of our proof of Theorem~\ref{T>=Q} we identify, define, and expose several other graph parameters which similarly characterize the complexity of a topology $G$, such as, $\movingcut{G}$, $\routing{G}$ and others. Many of these parameters have very different flavors. For example the $\movingcut{G}$ parameter can be seen as identifying crucial communication bottlenecks within a topology via a sequence of cuts. It is also known~\cite{haeupler2020network} to characterize the time needed to solve a simple multiple unicast communication problem which requires information to be sent between different sender-receiver pairs in the network. $\routing{G}$ relates to the same communication problem, but with the restriction that information is routed (without any coding) which, by \citet*{leighton1994packet}, is equivalent to the best congestion and dilation of paths connecting the sender-receiver pairs. We give precise definitions and further explanations for these and other equivalent universal lower bound parameters in the technical sections of this paper. We hope that they will help to further illuminate different aspects of the topology-complexity interplay.

        
\section{Shortcut Quality is a Universal Lower Bound}\label{sec:lower-bound}

In this section we present our proof of our universal lower bounds in terms of shortcut quality. In particular, this section is dedicated to proving the following theorem.
\begin{restatable}{theorem}{connectivitylb}\label{conn>=Q(G)}
  Let $\calA$ be any always-correct algorithm for spanning connected subgraph and let $T_{conn}(G) = \max_I T_{\calA}(G, I)$ denote the worst-case running time of $\calA$ on the network $G$. Then we have that:
  $$T_{conn}(G) = \tilde{\Omega}(\shortcut{G}).$$
\end{restatable}

We defer most proofs of this section to \Cref{sec:constructing-disjointness} and \Cref{sec:lower-bound-appendix}, focusing on a high-level overview.
We start by introducing \emph{disjointness gadgets}, which are pathological sub-graphs for distributed optimization, and outline their use in proving distributed lower bounds, in \Cref{sec:witnesses}. 
In order to obtain informative lower bounds from these gadgets, we then relate the worst such subgraph to the highest distance of any moving cut in $G$, $\movingcut{G}$, in  \Cref{sec:gadgets-in-general}. 
This is the technical meat of the paper, and \Cref{sec:constructing-disjointness} is dedicated to proving this relation. We then relate the obtained lower bounds to shortcut quality in \Cref{sec:beta-to-shortcuts}. Finally, we conclude with the proof of \Cref{conn>=Q(G)}, as well as discussions of its implications to other distributed problems, in \Cref{sec:putting-it-together}.

\subsection{Lower bound witnesses}\label{sec:witnesses}

In this section we define \emph{$\beta$-disjointness gadgets}, a structure that connects together information-theoretic bounds with higher-level distributed optimization problems like MST. The structure can be seen as a generalization of previous \emph{existential} lower bounds that show many distributed problems cannot be solved faster than $\tOmega(D+\sqrt{n})$ on a specific graph family~\cite{elkin2006unconditional,dassarma2012distributed,peleg2000near}. We argue that $\beta$-disjointness gadgets are the ``right'' way to generalize their approaches to arbitrary graphs.

\begin{definition}
  A $\beta$-\emph{disjointness gadget} $(P,T,\ell)$ in graph $G$ consists of 
  a set of vertex-disjoint paths $P\neq \emptyset$, each of length at least three; a tree $T\subseteq G$ which intersects each path in $P$ exactly at its endpoint vertices; and a moving cut of capacity strictly less than $|P|$ and distance $\beta$ with respect to the pairs $\{(s_i, t_i)\}_{i=1}^{|P|}$ of endpoints of paths $p_i\in P$.
\end{definition}
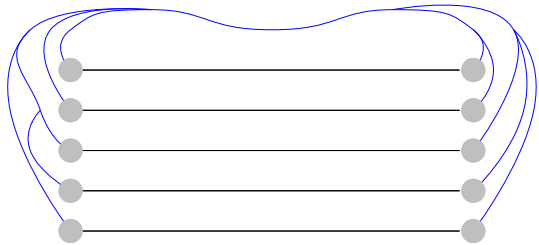
\begin{figure}[h]
  \centering

  \resizebox{0.45\textwidth}{!}{
  \begin{tikzpicture}[shorten >=1pt,-]
    \tikzstyle{vertex}=[circle,fill=black!25,minimum size=17pt,inner sep=0pt]
    \clip (-2,0.5) rectangle (12, 6.7);
    
    \foreach \name in {1,...,5} {
      \node[vertex] (s-\name) at (0,\name) {};
      \node[vertex] (t-\name) at (10,\name) {};
      \draw (s-\name) -- (t-\name);
    }
    \node[] (root) at (5,6) {};

    \draw [blue] plot [smooth, tension = 1] coordinates {
      (s-5) (10-10,6) (10-8,6.5) 
      (root) (8,6.5) (10,6) (t-5)};
    \draw [blue] plot [smooth, tension = 1] coordinates {(10,6) (10.5,5) (t-4)};
    \draw [blue] plot [smooth, tension = 1] coordinates {(8,6.5) (11,6) (t-3)};
    \draw [blue] plot [smooth, tension = 1] coordinates {(11,6) (11.25,4) (t-2)};
    \draw [blue] plot [smooth, tension = 1] coordinates {(11,6) (11.5,4) (t-1)};	

    \draw [blue] plot [smooth, tension = 1] coordinates {(10-8,6.5) (10-10.5,6) (s-4)};
    \draw [blue] plot [smooth, tension = 1] coordinates {(10-8,6.5) (10-11,6) (10-10.75,4) 
      (s-3)};
    \draw [blue] plot [smooth, tension = 1] coordinates {(10-10.75,4) (10-11,3) (s-2)};
    \draw [blue] plot [smooth, tension = 1] coordinates {(10-11,6) (10-11.5,4) (s-1)};

    \foreach \name in {1,...,5} {
      \node[vertex] (s-\name) at (0,\name) {};
      \node[vertex] (t-\name) at (10,\name) {};
      \draw (s-\name) -- (t-\name);
    }
    \node[] (root) at (5,7.25) {};    
  \end{tikzpicture}
  }

  \caption{A disjointness gadget's path and tree, given by straight and rounded blue lines, respectively.}
  \label{fig:disj-gadget}
\end{figure}

As we show, such disjointness gadgets are precisely the worst-case subgraphs which cause distributed verification (and optimization) to be hard. In particular, denoting by \wcsub{G} the highest value of $\beta$ for which there exists a $\beta$-disjointness gadget in $G$ (or zero, if none exists). This quantifies the most pathological subgraph in $G$. We prove the following.

\begin{restatable}{lemma}{witnesslb}\label{lem:lower-bound-witness}
  Let $\calA$ be any always-correct algorithm for spanning connected subgraph and let $T_{conn}(G) = \max_I T_{\calA}(G, I)$ denote the worst-case running time of $\calA$ on the network $G$. Then we have that:
  $$T_{conn}(G) = \tOmega(\wcsub{G} + D).$$
\end{restatable}

The non-trivial part of this lemma is the lower bound $T_{conn}(G) = \tOmega(\wcsub{G})$. Our proof of this bound (see \Cref{sec:lower-bound-appendix}) follows the approach implicit in \cite{dassarma2012distributed}. Broadly, we use a disjointness gadget $(P,T,\ell)$ to construct a subgraph $H$ determined by private $|P|$-bit inputs $x,y$ for the endpoints of the paths, such that $H$ is a spanning and connected subgraph of $G$ if and only if $\disj(x,y)=1$. 
Combined with \Cref{lemma:moving-cut-simulation}, this equivalence and the moving cut $\ell$ yield a lower bound on subgraph connectivity in a graph $G$ containing a $\beta$-disjointness gadget.
In following sections we show how to use this bound to prove lower bounds for this problem in \emph{any} graph $G$.

\subsection{Disjointness gadgets in any graph}\label{sec:gadgets-in-general}
The first challenge in deriving an informative lower bound on the time for spanning connected subgraph verification from \Cref{lem:lower-bound-witness} is that graphs need not contain disjointness gadgets. For example, as disjointness gadgets induce cycles, trivially no such gadgets exist in a tree. 
Consequently, for trees \Cref{lem:lower-bound-witness} only recreates the trivial lower bound of $\tOmega(D)$.

The following theorem implies that for any graphs where $\movingcut{G}$ is sufficiently larger than $D$, disjointness gadgets \emph{do} exist. 
More precisely, we prove the following theorem.

\begin{restatable}{theorem}{betadandbetac}\label{beta_D>=beta_C}
  For any graph $G$, 
  $$\wcsub{G}+D = \tilde{\Theta}(\movingcut{G}).$$
\end{restatable}

\Cref{beta_D>=beta_C} is the technical core of this paper, and \Cref{sec:constructing-disjointness} is dedicated to its proof. 
At a (very) high level, what we prove there is that, 
while disjointness gadgets do not always exist, some relaxation of them always does. 
In particular, we show that for any graph $G$ and set of connectable pairs $S$ in $G$, some relaxed notion of disjointness gadgets always exists for a subset $S'\subseteq S$ of size $|S'|=\Omega(|S|)$. The majority of \Cref{sec:constructing-disjointness} is dedicated to proving the existence of such relaxed disjointness gadgets.
We then show how to extend a moving cut of distance $\beta \ge 9D$ on $S$ to (strict) disjointness gadgets: construct a relaxed disjointness gadget on a large subset of $S$ (since $S$ are connectable), then clean-up the structure using $\beta \ge 9D$ to transform it to a (strict) $\beta$-disjointness gadget.

\subsection{Relating $\movingcut{G}$ to \shortcut{G}}\label{sec:beta-to-shortcuts}

So far we have shown that (up to polylog multiplicative terms and additive $O(D)$ terms), the time to solve subgraph connectivity is at least the length of the worst moving cut in $G$, which we denote by \movingcut{G}. More precisely, so far we proved that $$T_{conn}(G) \geq \tilde{\Theta}(\wcsub{G}+D) = \tilde{\Theta}(\movingcut{G}).$$
In this section we show that the above terms we have proven to be equivalent (up to polylog factors) are in turn equivalent to the graph's shortcut quality.

Indeed, by lemmas \ref{moving-cuts-communication} and \ref{LMR}, we have that $\movingcut{G} = \tilde{\Theta}(\pairshortcut{G})$. 
The following lemma proves an equivalence (up to polylog factors) between shortcut quality for pairs to the graph's shortcut quality (for parts). 

\begin{restatable}{lemma}{shortcutsrelation}\label{shortcut-relation}
	For any graph $G$, 
	$$\shortcut{G} = \tilde{\Theta}(\pairshortcut{G}).$$
\end{restatable}

Broadly, we use heavy-light decompositions \cite{sleator1983data} of spanning trees of parts, to show how to obtain shortcuts for parts by gluing together a polylogarithmic number of shortcuts for connected pairs. The overall dilation and congestion of the obtained shortcuts for the parts are at most polylogarithmically worse than those of the shortcuts for the pairs.
(See \Cref{sec:multicast-to-unicast} for proof.)

\subsection{Putting it all together}\label{sec:putting-it-together}

In this section we review our main result, whereby shortcut quality serves as a universal lower bound for the spanning connected subgraph problem, as well as numerous other problems.

\connectivitylb*
\begin{proof}
	Putting all the lemmas above together, we have
	\begin{align*}
	T_{conn}(G) & \geq \tilde{\Theta}(\wcsub{G} + D) & \Cref{lem:lower-bound-witness} \\
	& = \tilde{\Theta}(\movingcut{G}) & \Cref{beta_D>=beta_C} \\
	& = \tilde{\Theta}(\comm{G}) & \Cref{relationships} \\
	& = \tilde{\Theta}(\routing{G}) & \Cref{relationships} \\	
	& = \tilde{\Theta}(\pairshortcut{G}) & \Cref{relationships} \\		
	& = \tilde{\Theta}(\shortcut{G}) & \Cref{shortcut-relation} & \qedhere
	\end{align*}
\end{proof}

We note that the above proof entails a proof of \Cref{T>=Q}, as well as the equivalence between the number of tight universal lower bounds for our problems discussed in \Cref{sec:characterizatoins}.

Known reductions presented in \citet{dassarma2012distributed} extend the same universal lower bounds of \Cref{conn>=Q(G)} to numerous problems such as the MST, shallow-light tree, SSSP, min-cut and others. The reductions hold for both non-trivial approximation factors as well as randomized algorithms.

Since MST (and all above problems, for some approximation ratios) can be solved using $\tilde{O}(1)$ applications of partwise aggregation, \Cref{conn>=Q(G)} implies a similar $\tOmega(\shortcut{G})$ lower bound for the partwise aggregation problem. 
In \Cref{sec:shortcut-construction} we present a polytime algorithm matching this lower bound, resulting in polytime universally-optimal supported CONGEST algorithms for all problems studied in this paper.


\section{Constructing Disjointness Gadgets}\label{sec:constructing-disjointness}

The goal of this section is to prove \Cref{beta_D>=beta_C}, namely that the parameters $\wcsub{G} + D$ and $\movingcut{G}$ are equivalent, up to $\tO(1)$ factors.
At the heart of this proof is a lemma that extends a moving cut between connectable pairs to a disjointness gadget. The proof is fairly involved---which is why we first provide an abbreviated summary before presenting the formal argument.

\subsection{Technical overview}
At a high level, the proof defines two auxiliary structures: \emph{crowns} and \emph{relaxed disjointness structures} (both defined below). Given $k$ connectable pairs, we first show that one can always construct a crown on a large, $\Omega(k)$-sized, subset of these pairs. Next, we show that one can always construct relaxed disjointness gadgets on an $\Omega(k)$-sized subset of the pairs, by converting a crown to a relaxed disjointness gadget. 
Finally, to obtain a (strict) disjointness gadget, which also requires a moving cut, we show how to construct a disjointness gadget by considering a moving cut between connectable pairs and then adapting a (relaxed) disjointness structure on a large fraction of those pairs.

\paragraph{Crowns.} As crowns have a somewhat technical definition, we start by motivating their definition.

Given $k$ vertex-disjoint paths $\{\pth_i\}_{i=1}^k$ in $G$, we call the $p_i$'s \textbf{part-paths} and the indices $i$ \textbf{parts}. Suppose that the following (false) statement holds: ``One can always find a connected subgraph $T \subseteq E(G)$ that touches $\tOmega(k)$ part-paths''. By touching we mean that exactly one node and zero edges lie in the intersection of $\pth_i$ and $T$. Such a structure would be highly interesting---it would show, in an analogous fashion to \Cref{lem:lower-bound-witness}, that a moving cut on $k$ connectable pairs can be extended to a universal lower bound for the SSSP problem, by reducing disjointness to SSSP using this structure, and then appealing to \Cref{lemma:moving-cut-simulation}.

Unfortunately, the statement as written does not hold (e.g., when $G$ is a path), but can be relaxed in a way that is both true and does not break the reduction: we allow the intersection of $T$ and $\pth_i$ to be coverable by a sub-path of $\pth_i$ of length at most $D$---the graph diameter. This makes the path example trivial and changes the reduction up to a multiplicative constant and additive $\tO(D)$ factors, both of which are insignificant in the context of this paper. This relaxed structure is a (global) crown. However, constructing global crowns is challenging; our definition is essentially a local version of the above relaxed structure.
\begin{definition}[Crown]\label{def:crown}
  Let $\{\pth_i\}_{i=1}^k$ be a set of vertex-disjoint paths in a graph $G$ of diameter $D$. A triplet $(T, A, U)$, where $T \subseteq G$ is a connected subset of edges in $G$, and $U \subseteq A \subseteq [k]$ are two subsets of parts, is a \textbf{crown} if the following properties hold:
  \begin{enumerate}
  \item \label{crown-prop-u-vs-a} $|U| \ge \frac{1}{4}|A| + 2$.
  \item \label{crown-prop-f-touches-only-a} $T$ only intersects parts in $A$. More precisely, $V(T) \cap V(\pth_i) = \emptyset$ for all $i \in [k] \setminus A$.
  \item \label{crown-prop-f-path-covering} 
  $T$ intersects each part $i\in U$, and this (non-empty) intersection, $V(T)\cap V(\pth_i)$, can be covered by a single sub-path of $\pth_i$ of length at most $D$.
  \end{enumerate}
\end{definition}
We use the following crown terminology for expressiveness (see \Cref{fig:crown}). We say that part $i$ \emph{belongs} to crown $(T, A, U)$ if $i \in A$; $i$ is \emph{useful} if $i \in U$; $i$ is \emph{sacrificial} if $i \in A \setminus U$ (note: $A$ stands for ``all'', $U$ stands for ``useful''). While not a part of the definition, for the crowns we consider, the sacrificial parts will always be fully contained in $T$, i.e., if $i \in A \setminus U$, then $E(\pth_i) \subseteq T$.
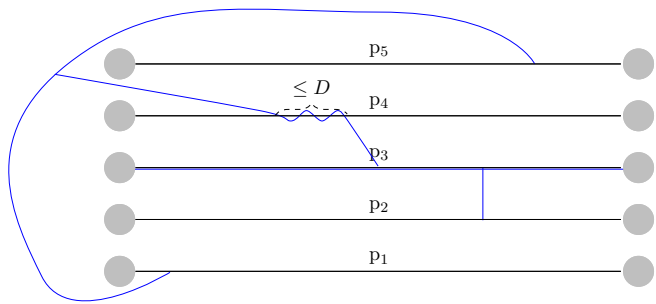
\begin{figure}[h]
  \centering

  \resizebox{0.6\textwidth}{!}{
  \begin{tikzpicture}[shorten >=1pt,-]
    \tikzstyle{vertex}=[circle,fill=black!25,minimum size=17pt,inner sep=0pt]
    \clip (-2.5,0) rectangle (12, 6.7);
    
    \foreach \name in {1,...,5} {
      \node[vertex] (s-\name) at (0,\name) {};
      \node[vertex] (t-\name) at (10,\name) {};
      \draw (s-\name) -- node[above]{$\pth_\name$} (t-\name);
    }
    \node[] (root) at (5,6) {};


    \draw [blue] plot [smooth] coordinates {(7, 3) (7, 1.95)};
    
    \draw [blue] plot [smooth] coordinates {(-1.25, 4.8) (2.7, 4.1) (3.3, 3.9) (3.6, 4.1) (3.9, 3.9) (4.2, 4.1) (4.4, 3.9) (5, 3)}; 

    \draw [dashed, decorate, decoration={brace, amplitude=7pt}] (3, 4) -- (4.4,4) node[midway, yshift=15pt] {$\le D$} ;  


    \draw [blue] plot [smooth] coordinates {(0,2.97) (10,2.97)};
    
    \draw [blue] plot [smooth, tension = 1] coordinates {(8,5) (root) (10-11,5) (10-11.5,0.9) (1,1) }; 

    \foreach \name in {1,...,5} {
      \node[vertex] (s-\name) at (0,\name) {};
      \node[vertex] (t-\name) at (10,\name) {};
      \draw (s-\name) -- (t-\name);
    }
    \node[] (root) at (5,7.25) {};    
  \end{tikzpicture}
  }
  \caption{A crown $(T, \{1, 2, 3, 4, 5\}, \{1, 2, 4, 5\})$. $T$ is depicted in blue. Note that part $3$ is sacrificial, hence is contained in $T$. The intersection of $T$ and $\pth_4$ is covered by a sub-path of length at most $D$. Intersections with other useful part-paths are covered by a trivial sub-path of length $0$.}
  \label{fig:crown}
\end{figure}

We say that two crowns $(T_1, A_1, U_1)$ and $(T_2, A_2, U_2)$ are \emph{disjoint} if $A_1 \cap A_2 = \emptyset$ (though $T_1$ and $T_2$ can intersect arbitrarily). A set of crowns is disjoint if every pair is disjoint. The definition of crowns implies that a set of disjoint crowns can be easily merged into a single crown, as follows. Consider a path $q$ between two crowns $(T_1, A_1, U_1)$ and $(T_2, A_2, U_2)$ that does not touch any other crown (such a path must exist). Merge the two crowns via $(T_1 \cup T_2 \cup E(q), A_1 \cup A_2, U_1 \cup U_2)$ and declare all part-paths that $q$ touches sacrificial (at most 2). All crown properties continue to hold.

We take a step back and compare the crown definition with the above SSSP motivation. In order for the reduction to work, it is sufficient to prove that ``for every $k$ part-paths one can construct a crown $(T, A, U)$ with $|A| \ge \Omega(k)$' (while only useful parts have a bounded intersection and can be used in the reduction, Property \ref{crown-prop-u-vs-a} of crowns relaxes this requirement to the previous statement). Combining this with the merging property, it is sufficient to show that ``for every set of $k$ part-paths one can partition $\Omega(k)$ of them into crowns''. We argue this is true by considering two sub-cases, 
phrased in terms of the \emph{contraction graph} $R$, whose vertices are parts $[k]$ whose edges $\{i, j\}$ corresponds to pairs of indices with a path (in $G$) between $p_i$ and $p_j$ that does not touch any other part-path. 
In the \emph{(high-degree case)}, $\Omega(k)$ part-paths are adjacent to at least three other part-paths, and in the \emph{(low-degree case)} when almost all part-paths are adjacent to at most two other part-paths.

\paragraph{High-degree case.} We consider the case when $\Omega(k)$ parts have $R$-degree at least three. Pick a part $i$ and seed (i.e., start) a crown from it: declare $i$  sacrificial and all of its $R$-neighbors useful. Continue growing the crown as long as a useful part $j$ with at least two ``unused'' $R$-neighbors exists, in which case declare $j$ sacrificial and its neighbors to the useful parts to the crown. When we cannot grow the crown anymore we add it to the collection, delete the used parts from the graph and repeat the process with another seed $i$ that has three unused neighbors. It is easy to argue that $\Omega(k)$ parts belong to some crown: an unused part $i$ of $R$-degree three or more has at least one neighboring part $j$ in some crown (otherwise $i$ would start its own crown), and $j$ is a useful part in its crown (otherwise $j$'s crown would absorb $i$). However, useful parts belonging to some crown can have only one unused neighbor (otherwise they would grow a crown), hence we can charge each unused part to a unique useful part in a crown. This shows that only a small fraction of parts of degree at least three can be unused.

\paragraph{Low-degree case.} We illustrate our techniques on the case when \emph{all} vertices in $R$ have degree two (the formal proof handles the case when \emph{most} nodes have this property). We decompose $R$ into cycles and paths and construct crowns on a constant fraction of each. For an $R$-path, we show one can construct a crown on any three consecutive parts $(a, b, c)$, thereby proving one can construct a crown on a constant fraction of a $R$-path (of length at least three). Fix $(a, b, c)$ and consider the shortest path $f$ from any node in $p_a$ to any node in $p_c$; the length of $f$ is at most the diameter of $G$, namely $D$. Note that $f$ intersects $p_b$, but if the intersection is coverable by a sub-path of length at most $D$ then one can make a crown $(E(f), \{a, b, c\}, \{a, b, c\})$. If this is not the case, we can ``replace'' the part of $p_b$ between the first and last intersection with $f$ with the appropriate part of $f$, forcing the intersection to be coverable by a short sub-path; this proves one can construct a crown on $\{a, b, c\}$. Cycles can be handled similarly. Combining both cases, we conclude that for every set of $k$ part-paths, there exists a crown on some $\Omega(k)$-sized subset of them.

\paragraph{Converting a crown into a relaxed disjointness gadget.} Relaxed disjointness gadgets are an intermediate step between crowns and disjointness gadgets. Relaxed disjointness gadgets require both endpoints of part-paths to be included in $T$ (like a disjointness gadget), but also allow for three exceptional sub-paths of length at most $D$ on each part-path (unlike crowns that allow only one). Moreover, relaxed disjointness gadgets do not require a moving cut.
\begin{definition}\label{def:relaxed-disjointness}
  A \textbf{relaxed disjointness gadget} $(P,T)$ in graph $G$ of diameter $D$ consists of vertex-disjoint paths $P=\{\pth_i\}_{i=1}^k$ and a connected subset of edges $T\subseteq E(G)$ which intersects all paths at their endpoint vertices, and such that for each $i\in |P|$, $V(T)\cap V(\pth_i)$ can be covered by at most three sub-paths of $\pth_i$ of length at most $D$. We say that the \textbf{endpoints} of $(P,T)$ are the endpoints of $\{\pth_i\}_{i=1}^k$.
\end{definition}
\begin{figure}[h]
  \centering
	\vspace{-1.25cm}
  \resizebox{0.65\textwidth}{!}{
  \begin{tikzpicture}[shorten >=1pt,-]
    \tikzstyle{vertex}=[circle,fill=black!25,minimum size=17pt,inner sep=0pt]
    \clip (-2.5,0) rectangle (12, 6.7);
    
   \foreach \name in {1,...,3} {
      \node[vertex] (s-\name) at (0,\name) {};
      \node[vertex] (t-\name) at (10,\name) {};
      \draw (s-\name) -- node[pos=0.35, above]{$\pth_\name$} (t-\name);
    }
    \node[] (root) at (5,5) {};

    \draw [blue] plot [smooth] coordinates {(root) (5.2, 3.1) (5.4, 2.9) (5.8, 3.1) (6.1, 2.9) (6.5, 3.1) (6.6, 2.9) (9, 2) (10, 2)}; 

    \draw [dashed, decorate, decoration={brace, amplitude=7pt}] (5.2,3) -- (6.6,3) node[midway, yshift=15pt] {$\le D$} ;  

    \draw [dashed, decorate, decoration={brace, amplitude=7pt}] (0,3) -- (1.5,3) node[midway, yshift=15pt] {$\le D$} ;  

    \draw [dashed, decorate, decoration={brace, amplitude=7pt}] (8.5,3) -- (10,3) node[midway, yshift=15pt] {$\le D$} ;  

    \draw [blue] plot [smooth] coordinates {(5.4, 2.9) (0.9, 2) (0, 2)}; 
    
    \draw [blue] plot [smooth] coordinates {(root) (1.5, 3.1) (1, 2.9) (0.5, 3.1) (0, 3)}; 

    \draw [blue] plot [smooth] coordinates {(root) (8.5, 3.1) (9, 2.9) (9.5, 3.1) (10, 3)}; 

    \draw [blue] plot [smooth] coordinates {(root) (10.5,4.5) (11,2.5) (t-1)};	
    \draw [blue] plot [smooth] coordinates {(root) (10-11,4.5) (10-11.5,2.5) (s-1)};

    \foreach \name in {1,...,3} {
      \node[vertex] (s-\name) at (0,\name) {};
      \node[vertex] (t-\name) at (10,\name) {};
      \draw (s-\name) -- (t-\name);
    }
    \node[] (root) at (5,7.25) {};    
  \end{tikzpicture}
  } 
  \vspace{-0.25cm}
  \caption{A relaxed disjointness gadget. The edges of the paths and $T$ are horizontal and blue lines, respectively.
  The intersection of $\pth_3$ and $T$ is covered by three sub-paths of length at most $D$. 
}
  \label{fig:relaxed-disj-gadget}
\end{figure}
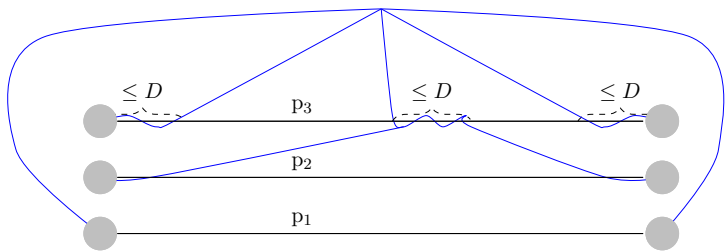
We show that ``for every crown $(T, A, U)$ on $k$ parts, one can construct a relaxed disjointness gadget on $\Omega(k)$ parts''. The proof sketch is fairly simple: we need to ``connect'' the endpoints $\{s_i, t_i\}$ of $\Omega(k)$ parts to $T$. Fix some $i$ and find the shortest path from each endpoint of $p_i$ to $T$. If both $s_i$ and $t_i$ can connect to $T$ without intersecting another part-path $p_j$, we include the connections in $T$ and continue to the next $i$. If this is not the case, we can either connect $i$ and sacrifice (at most two) other part-paths, or not include $i$. One can always choose a constant fraction of $i$'s that can be simultaneously connected to $T$: consider the ``interference graph'' that has a directed edge $i \to j$ if connecting $i$ sacrifices $j$. The graph has out-degree at most two, hence total (in+out) average degree at most four; by Turan's theorem, there exists an independent set of size $k/5$, which can be simultaneously connected and, conceptually, we are done. There are significant technical challenges that lie beyond this simple sketch: the connections from $s_i$ and $t_i$ could arbitrarily intersect $p_i$ (i.e., not be coverable by a sub-path of length at most $D$). This can be dealt with by path-replacement strategies such as the ones in the low-degree crown construction. However, a replacement could start to intersect other connections which are not accounted for in the interference graph, requiring special care.

\paragraph{Converting a relaxed disjointness gadget to a disjointness gadget.} To prove \Cref{beta_D>=beta_C} we need to bring together moving cuts and relaxed disjointness structures. Suppose we are given a moving cut $\ell$ between $k$ connectable pairs. As previously argued, we can construct a relaxed disjointness gadget $(\{p_i\}_i, T)$ on $\Omega(k)$ of such pairs. The moving cut $\ell$ ensures that the $\ell$-length of each $p_i$ is at least $\beta$ (assume $\beta \ge 9D$). One can assume WLOG that $\ell(e) = 1$ for each edge on each part-path. For each $p_i$ we exclude the sub-paths that intersect $T$ (total exclusion has $\ell$-length and length at most $3D$); this partitions $p_i$ into two sub-paths, one of which (denoted by $p'_i$) has $\ell$-length at least $(\beta - 3D) / 2$. Adding $\bigcup_i p_i \setminus p'_i$ to $T$, we maintain the property that the endpoints of $p'_i$ connect to $T$ and that $p'_i$ does not intersect $T$ internally. This almost completes the disjointness structure construction. Denote the endpoints of $p'_i$ with $s'_i, t'_i$. The final property we need to ensure is that $\dist_\ell(s'_{\pmb{i}}, t'_{\pmb{j}}) \ge \tOmega(\beta)$ for all $i,j$; instead, we only ensured that $\dist_\ell(s'_i, t'_i) \ge \tOmega(\beta)$ for all $i$. 
However, 
going between one such distance guarantee and the other can be done via a structural lemma from prior work by losing an extra $O(\log n)$ factor in $\beta$~\cite[Lemma 2.5]{haeupler2020network}. This completes the full proof.

\subsection{Constructing crowns}

In this section we show that large crowns exist in all graphs (see \Cref{def:crown}). We first describe the crown merging procedure and show how to construct them. The construction is partitioned into the high-degree case and the low-degree case. We start by defining some notation we will be using throughout this section.

\paragraph{Notation specific to \Cref{sec:constructing-disjointness}.} We denote the disjoint union via $A \sqcup B$, where we guarantee $A \cap B = \emptyset$. For a graph $H$, we denote by $V(H)$ and $E(H)$ the vertices and edges of $H$, respectively. For a set of edges $T \subseteq E(H)$, we denote by $V(T):=\bigcup_{e \in T} e$ the set of endpoints of edges in $T$. Similarly, for a path $p$ we denote by $V(p)$ the set of nodes on the path (including its endpoints).

We denote the degree of a vertex $v$ in graph $H$ by $\deg_H(v)$, and its neighborhood in $H$ by $\Gamma_H(v)$. We extend the neighborhood definition to sets of vertices $X \subseteq V(H)$, letting $\Gamma_H(X) := \bigcup_{v \in X} \Gamma_H(v)$ (i.e., the set of vertices in $V(H)$ that have a neighbor in $X$). Furthermore, we denote the inclusive and exclusive neighborhoods by $\Gamma_H^{+}(X) := \Gamma_H(X) \cup X$ and $\Gamma_H^{-}(X) := \Gamma_H(X) \setminus X$.

A \emph{walk} in $H$ is a sequence of vertices $w = (w_0, w_1, \ldots, w_\ell)$, where $w_i \in V(H)$ and $\{w_i, w_{i+1}\} \in E(H)$ for all $0\leq i\leq \ell-1$. The indices $i$ are often called \emph{steps}. The length of the walk is denoted by $|w| := \ell$. Given two steps $i \le j$ we denote by $w_{[i,j]}$ the \emph{subwalk} $(w_i, w_{i+1}, \ldots, w_j)$. Furthermore, we define $E(w)$ as the set of edges the walk traverses over and we define $V(w) := \bigcup_{i=0}^\ell \{ w_i\}$.
Given two nodes $u, v \in V(H)$ the distance $\dist_H(u, v)$ is the length of the shortest walk between them. We extend the distance definition to sets $X, Y \subseteq V(H)$ in the natural way, letting $\dist_H(X, Y) := \min_{x \in X, y \in Y} \dist_H(x, y)$.

We define walk \emph{clipping}. For a walk $w$ in graph $H$, and two sets $A, B \subseteq V(H)$, we denote the subwalk of $w$ from the last step $i := \max\{i \mid w_i\in A\}$ corresponding to a vertex in $A$ to the first step $j := \min \{j \mid w_j \in B, j > i\}$ corresponding to a vertex in $B$ by $\clip(w, A, B) := w_{[i, j]}$.
We note that the clipping operation is well defined only when after each step $i$ where $w_i \in A$ there exists a step $j \ge i$ such that $w_j \in B$. This will always be the case when we use this operation. (This is, for instance, true if $w_0 \in A$ and $w_\ell \in B$.)

The definition of crowns allows for disjoint crowns to merged rather directly, by sacrificing at most one part in each crown in order to merge the trees. See \Cref{sec:constructing-disjointness-appendix} for a full proof.

\begin{restatable}{lemma}{crownmerging}(Crown merging)\label{crown-merging}
  Let $\{(T_i, A_i, U_i)\}_{i \in I}$ be a set of disjoint crowns of vertex-disjoint paths $\{ \pth_i \}_{i=1}^k$. Then there exists a single crown $(T_*, A_*, U_*)$ of $\{ \pth_i \}_{i \in A_*}$ such that $A_* = \bigsqcup_{i \in I} A_i$.
\end{restatable}

\subsubsection{Contraction graph}

In this section we define the contraction graph which we will extensively used in the crown construction.

\begin{definition}[Contraction graph]
  A \textbf{contraction graph} $R$ with respect to $k$ fixed vertex-disjoint paths $\{ \pth_i \}_{i=1}^k$ has vertex set $V(R)=[k]$. 
  The edges of $R$ consist of all pairs $\{i, j\}$ such that there exists a simple path $(v_0, v_1, \ldots, v_\ell)$ in $G$ where $v_0 \in V(\pth_i)$, $v_\ell \in V(\pth_j)$,  and the internal nodes do not belong to any part-path, i.e. $\bigcup_{x=1}^{\ell-1} \{v_x\} \cap \bigcup_{y=1}^k V(\pth_y) = \emptyset$. We denote this path by $E_G(\{i, j\})$.
\end{definition}

Closely related to the contraction graph, is the \textbf{projection} $\pi_{G \to R}$, which is a mapping between walks $w = (v_0, v_1, \ldots, v_\ell)$ in $G$ and walks in $R$. We first fix a $v \in V(G)$ and define $\pi_{G \to R}(v) = x$ if $v \in \pth_x$ and $\pi_{G \to R}(v) = \bot$ if $v$ does not belong to any part-path (note that our definition is well-defined because the part-paths are disjoint). We now extend the function to walks $w = (v_0, \ldots, v_\ell)$ on $G$. We define $w' := (\pi_{G \to R}(v_0), \pi_{G \to R}(v_1), \ldots, \pi_{G \to R}(v_\ell))$. Furthermore, we let $w''$ be the sequence $w'$ with elements $\bot$ removed, and consecutive duplicate elements merged down to a single element. For instance, if $w' = (5, 1, \bot, 1, 1, \bot, 2, 2, \bot, 2, 3\}$, then $w'' = (5, 1, 2, 3)$. The projection $\pi_{G \to R}$ maps $w \mapsto w''$. It is easy to see that $w''$ is a walk on $R$.

\begin{observation}\label{contraction-graph-is-connected}
  The contraction graph $R$ is connected.
\end{observation}
\begin{proof}
  Let $a, b \in V(R)$ be parts and let $a', b'$ be arbitrary nodes satisfying $a' \in \pth_a, b' \in \pth_b$. Since $G$ is connected, there exists a walk $w$ connecting $a'$ and $b'$. Projecting the walk to $R$, $\pi_{G \to R}(w)$ is a walk in $R$ connecting $a$ and $b$, proving the claim.
\end{proof}

\subsubsection{Constructing crowns in the high-degree case}

In this section we consider the case where a large fraction of parts either have $R$-degree of at least three or a neighbor with this property. In this case, combining with the crown-merging \Cref{crown-merging}, we can successfully construct a global crown on a constant fraction of parts.

\begin{lemma}\label{crowns-high-degree}
  Let $R$ be a contraction graph $R$, and let $H = \{ v \in V(R) \mid \deg_R(v) \ge 3 \}$. If $|\Gamma_R^{+}(H)| \ge \frac{1}{10} k$, then there exists a collection of disjoint crowns $\mc{C} = \{ (T_i, A_i, U_i) \}_{i \in I}$, with $\sum_{i \in I} |A_i| \ge \frac{1}{70} k$.
\end{lemma}
\begin{proof}
  We perform an iterative procedure to find the required collection. Let be $\mathcal{C}$ an (initially empty) list of disjoint valid crowns. Initially, we let the set of ``available parts'' be $L \gets V(R)$ (i.e., all parts).

  The procedure repeatedly ``seeds'' a new crown and then it proceeds to ``grow'' it until no longer possible. We describe the seeding procedure. Find a part $v \in L$ with at least three C-neighbors in $L$ (i.e. $|\Gamma_R(v) \cap L| \ge 3$); the procedure stops when no viable $v$ can be found. We seed a new crown $(T, A, U)$ from $v$: We define $A \gets \Gamma^{+}_R(v) \cap L$, and $L \gets L \setminus A$. The part $v$ is sacrificial and the $R$-neighbors of $v$ are useful, i.e., $U \gets A \setminus \{ v \}$. We assign $T \subseteq E(G)$ to the union of $E(\pth_v)$ (edges in $G$ corresponding to the part-path $\pth_v$) and $\bigcup_{w \in U} E_G(\{v, w\})$ (all of the simple paths corresponding to $R$-edges between $v$ and its useful neighbors). This completes the seeing portion.

  We proceed to ``grow'' the new crown. Find a useful part $w \in U$ in the current crown with at least two $R$-neighbors in $L$ (i.e., $|\Gamma_R(w) \cap L| \ge 2$). Let $X \gets \Gamma_R(w) \cap L$. We add these neighbors to the crown and sacrifice $w$: $A \gets A \cup X$, $U \gets U \setminus \{w \} \cup X$ and $L \gets L \setminus A$. We add to $T$ the part-path corresponding to $w$ and the simple paths connecting $w$ to the new useful parts $T \gets T \cup E(\pth_w) \cup \bigcup_{x \in X} E_G(\{w, x\})$. We repeat the growing step until no viable $w$ can be found. At that point we add $(T, A, U)$ to $\mathcal{C}$ and the seeding procedure is restarted.

  Suppose the above procedure yields crowns $\mathcal{C} = \{(T_i, A_i, U_i)\}_{i \in I}$. We proceed to analyze it. Note that by construction $V(R) = L \sqcup \left(\bigcup_{i \in I} A_i\right)$ where $\sqcup$ denoted disjoint union. Let us define $H = H_{in} \sqcup H_{out}$ (disjoint union) where $H_{in}$ are the parts of $H$ that are inside some crown ($H_{in} = H \cap (\bigcup_i A_i)$) and $H_{out}$ outside $H_{out} = H \cap L$, therefore we have $|H| = |H_{in}| + |H_{out}|$.

  Subclaim: we argue that $|\bigcup_{i \in I} A_i| \ge \frac{1}{2}|H|$. However, we first establish that $|H_{out}| \le |\bigcup_{i \in I} A_i|$ via a charging argument from $H_{out}$ to $\bigcup_{i \in I} A_i$. Let $h \in H_{out}$. Since $h \in L$ it must have at least one neighbor $v$ that belongs to some crown $(T, A, U)$ (choose an arbitrary $v$ if multiple neighbors fit the condition); if this were not the case then $|\Gamma_R(h)| \ge 3$ would ensure that a crown would be seeded from $h$. Furthermore, it must be that $v \in U$, since sacrificial parts absorb all of their available neighbors when joining a crown. We ``charge'' $h$ to $v$. On the other hand, fix a part $c$ in some crown. Note that at most one $h \in H_{out}$ can charge itself to $c$: if two parts in $H_{out}$ charge themselves to $c$, we would grow the crown via $c$ (note that $v$ must be a useful part in its crown if anyone charges to it since sacrificial parts absorb their available neighbors). Since each $h \in H_{out}$ can be charged to a unique part in $\bigcup_{i \in I} A_i$ we conclude that $|\bigcup_{i \in I} A_i| \ge |H_{out}|$. We complete the subclaim by noting that $|\bigcup_{i \in I} A_i| \ge |H_{in}| = |H| - |H_{out}| \ge |H| - |\bigcup_{i \in I} A_i|$, which can be rewritten as $|\bigcup_{i \in I} A_i| \ge \frac{1}{2}|H|$.

  Let use define $M := \Gamma_R(H) \setminus H$. Furthermore, partition $M = M_{in} \sqcup M_{out}$ where $M_{in} = M \cap (\bigcup_{i \in I} A_i)$ and $M_{out} = M \cap L$. Subclaim: we prove that $|H| \ge \frac{1}{2} |M_{out}|$ by charging elements of $M_{out}$ to $H$ in such a way that at most two elements of $M_{out}$ get charged to the same $H$. Fix $v \in M_{out} \subseteq \Gamma_R(H) \cap L$. By definition of $\Gamma_R(H)$, there exists $h \in \Gamma_R(v) \cap H$; we charge $v$ to $h$. We now argue that at most two different parts can charge the same $h$. First, note that a part $h \in H$ that is being charged to by at least one part must be either available (i.e., in $L$), or a useful part in some crown (since sacrificial parts in a crown consume all of their available neighbors). Assume, for the sake of contradiction, that a node $h$ is charged more than two times. Then the construction would either seed a crown (if $h \in L$) or grow a crown via $h$. We conclude that $|H| \ge \frac{1}{2} |M_{out}|$.

  We finalize the proof using the two subclaims and $M_{in} \subseteq \bigcup_{i \in I} A_i$:
  \begin{align*}
    \frac{1}{10}k \le |\Gamma_R^+(H)| & = |H| + |M_{in}| + |M_{out}| \\
                                      & \le |H| + |\bigcup_{i \in I} A_i| + 2 |H| \le (2 + 1 + 4)|\bigcup_{i \in I} A_i| = 7 |\bigcup_{i \in I} A_i|.
  \end{align*}
  This can be rewritten as $|\bigcup_{i \in I} A_i| = \sum_{i \in I} |A_i| \ge \frac{1}{70} k$.
\end{proof}

\subsubsection{Constructing crowns in the low-degree case}

In this section we consider the second case, where a large fraction of parts have $R$-degree at most two.

We start by introducing the notion of \textbf{minimal part-paths} with endpoints $S := \{(s_i, t_i)\}_{i=1}^k$. A set of vertex-disjoint paths $\{\pth_i\}_{i=1}^k$ with endpoints $S$ is called \emph{minimal} if they minimize $\sum_{i \in [k]} |\pth_i|$ among all such sets with endpoints $S$.

\begin{lemma}\label{crowns-low-degree}
  Let $R$ be a contraction graph with respect to some minimal part-paths, and let $H = \{ v \in V(R) \mid \deg_R(v) \ge 3 \}$. If $|\Gamma_R^{+}(H)| \le \frac{1}{10} k$, then there exists a collection of disjoint crowns $\mc{C} = \{ (T_i, A_i, U_i) \}_{i \in I}$ with $\sum_{i \in I} |A_i| \ge \frac{3}{100} k$.
\end{lemma}

The rest of this section will function as a proof of \Cref{crowns-low-degree}. This will allow us to break down its complexity into smaller pieces.

We define $M_0 := \Gamma_R^{+}(H)$, $M_i := \Gamma^{-}(M_{i-1})$ for $i \in \{1, \ldots, 8\}$, and $M := \bigcup_{i=0}^8 M_i$. In particular, for each $m \in M_i$ ($i \in \{0, \ldots, 8\}$) we have that $\dist_R(m, H) = i+1$.

\begin{claim}
  $|M_i| \le |M_{i-1}|$ for $i = 1, \ldots, 8$.
\end{claim}
\begin{proof}
  By definition, each $m \in M_i$ is a neighbor of some $m' \in M_{i-1}$ of degree at most two (there might be two choices for $m'$, but in this case we choose arbitrarily). We charge $m$ to $m'$. At most one part can be charged to $m'$: otherwise $m'$ could not be with distance $i$ from $H$ since then $i = \dist_G(m', H) = \dist_G(\Gamma_G(m'), H) + 1 = i+2$.
\end{proof}

We now define $O := V(R) \setminus M$. We have that $|M_8| \le \ldots \le |M_0| \le \frac{1}{10} k$, implying that $|M| \le \frac{9}{10} k$, and giving us that $|O| \ge k - \frac{9}{10}k = \frac{1}{10} k$.%

Consider the induced subgraph $R' := R[V(R) \setminus H]$. $R$ is connected (\Cref{contraction-graph-is-connected}). By definition of $H$, all parts in this subgraph have $R$-degree at most two, hence we can decompose $R'$ into paths. Note that cycles are not allowed since they would be an isolated connected component of $R$. There are two special cases: $k \le 2$ which is a trivial case, or the entire graph being a cycle in which case we redefine $R' = R[V(R) \setminus \{v\}]$ where $v$ is an arbitrary part and the rest of the proof remains unchanged. Let $\{q_i\}_{i \in J}$ be the collection of paths in $R'$ of length at least $9$ (discard all shorter paths), in other words, $|V(q_i)| \ge 10$.

\begin{claim}
  $\sum_{i \in J} |V(q_i)| \ge \frac{1}{10} k$.
\end{claim}
\begin{proof}
  We argue that $O \subseteq \bigcup_{i \in J} V(q_i)$ which would, together with disjointness of $V(q_i) \cap V(q_j) = \emptyset$, imply that $\sum_{i \in J} |V(q_i)| = |\bigcup_{i \in J} V(q_i)| \ge |O| \ge \frac{1}{10} k$. Fix a part $v \in O$. By definition, we have that $\dist_R(v, H) \ge 10$, hence the $R'$-connected-component that includes $v$ must be a path of length at least $9$, leading to the required conclusion.
\end{proof}

We introduce some notation: we say that a crown $(T, A, U)$ is \emph{supported} on a subset of parts $X \subseteq [k]$ if $A \subseteq X$. Similarly, for a path $f$ in $R$ we say that a $(T, A, U)$ is supported on the path if $A \subseteq V(f)$ (note that $V(f) \subseteq V(R) \subseteq [k]$).

\begin{claim}\label{crowns-on-a-path}
  Let $f$ be a simple path in $R$ with $|V(f)| \ge 10$ and the extra condition that for all $v \in V(f)$ we have $\deg_R(v) \le 2$. Then, there exists a collection of disjoint crowns supported on $f$ such that $\mathcal{C} = \{ (T_i, A_i, U_i) \}_{i \in I'}$ with $\sum_{i \in I'} |A_i| \ge \frac{3}{10} |V(f)|$.
\end{claim}
\begin{proof}
  Suppose that $f = (f_0, f_1, f_2, \ldots, f_{\ell})$ with $\ell = |V(f)| - 1 \ge 9$, where $f_i \in V(R)$ are parts. Let $x = 3 \lceil \ell / 9 \rceil - 1$. We note that $f_0, f_x$ and $f_{2x}$ are all valid parts since $2x = 6\lceil \ell / 9 \rceil - 2 \le 6 \ell / 9 + 6 - 2 = \frac{2}{3} \ell + 4 \le \ell$, where the last inequality holds for $\ell \ge 12$ and can be manually checked for $\ell = 9, 10, 11$ (giving $2x = 4, 10, 10$, respectively). Furthermore, we note that $x + 1 \ge \frac{3}{10} |V(f)|$ for $|V(f)| \ge 10$ since $x + 1 = 3\lceil\ell/9\rceil \ge \ell/3 = \frac{|V(f)| - 1}{3} = \frac{0.9|V(f)| + 0.1|V(f)| - 1}{3} \ge \frac{0.9|V(f)|}{3} = \frac{3}{10} |V(f)|$.

  Subclaim: there exists a simple path $w'$ in $G$ (note: not $R$) of length at most $D$ whose projection $\pi_{G \to R}$ is supported on $f$ and intersects exactly $x+1$ part-paths. More precisely, $V(\pi_{G \to R}(w')) \subseteq V(f)$ and $|V(\pi_{G \to R}(w'))| = x + 1$. Let $a, b \in V(G)$ be arbitrary nodes on the part-paths corresponding to $f_x$ and $f_{2x - 1}$, respectively. Since the diameter of $G$ is at most $D$, there exists a simple path $w = (w_0 = a, w_1, \ldots, w_d = b)$ in $G$ from $a$ to $b$ of length at most $d \le D$. Consider $w' := \clip(w, V(\pth_{f_x}), V(\pth_{f_0}) \cup V(\pth_{f_{2x}}))$; the clipping is well-defined because the walk starts on $\pth_{f_x}$ and ends on $f_{2x}$. Furthermore, we claim that $\pi_{G \to R}(w')$ is supported on $f$: either $\pi_{G \to R}(w')$ is supported on $f$ and the statement is immediate or $w'$ exists $f$, in which case it must go through $\pth_{f_{0}}$ or $\pth_{f_{2x}}$; in both cases the path is clipped only to the subwalk supported on $f$. We conclude that $f$ intersects exactly $x+1$ part-paths since it starts on $\pth_{f_x}$ and ends in $\pth_{f_{0}}$ or $\pth_{f_{2x}}$. It must cover all parts in between due to the projection $\pi_{G \to R}$ being a walk on $R$ and cannot cover anything outside do to clipping. This completes the subclaim.

  Let $\pi_{G \to R}(w')$ intersect exactly $f' := \{f'_0, f'_1, \ldots, f'_x\}$, where $f'_i$ and $f'_{i+1}$ are consecutive parts on the path $f$ and $w'$ start (resp., end) on a node corresponding to $\pth_{f'_0}$ (resp., $\pth_{f'_x}$).

  We construct crowns on $f'$ by partitioning $f'$ into triplets $(f'_0, f'_1, f'_2), (f'_3, f'_4, f'_5), \ldots, (f'_{x-2}, f'_{x-1}, f'_x)$ (note that $3 \mid x+1$, i.e., $x+1$ is divisible by $3$). Fix a triplet $(f'_i, f'_{i+1}, f'_{i+2})$ for $3 | i$ and construct a crown on it as follows.

  Let $w^i := \clip(w', V(\pth_{f'_i}), V(\pth_{f'_{i+2}}))$ and note that $|w^i| \le D$. The clipping is well defined because since every subwalk of $w'$ that passes through $\pth_{f'_i}$ must eventually cross over $\pth_{f'_{i+2}}$ before ending at $\pth_{f'_x}$. Let $\tilde{w}^i$ be the subwalk of $w^i$ between the first and last occurrence of a node that is on $\pth_{f'_{i+1}}$; this is well-defined since $w^i$ intersects $\pth_{f'_{i+1}}$. Let $u, v$ be the first and last node of $\tilde{w}^i$. Note that, because of path minimality, $\pth_{f'_{i+1}}$ connects $u$ and $v$ via some shortest path in the subgraph of $G$ that excludes the other part-paths (otherwise we could shorten the $\pth_i$); call this sub-path $q$. This is because each sub-path of the shortest path (e.g., $q$) is also a shortest path. On the other hand, $\tilde{w}^i$ is some walk connecting $u$ and $v$ that does not touch any other part-path outside of $V(\pth_{f'_{i+1}})$. Therefore, the length of $\tilde{w}^i$ cannot be smaller than $q$. This allows us to swap the subwalk of $w^i$ corresponding to $\tilde{w}^i$ with $q$ without increasing the length of the walk $w^i$. In the reminder we assume we have performed this replacement.

  We construct a crown $(T, A, U)$ by assigning $T \gets w^i$, $A \gets U \gets \{f'_i, f'_{i+1}, f'_{i+2}\}$ and verify it is a valid crown. $T$ is connected since it is a walk. Property \ref{crown-prop-u-vs-a} is satisfied since $3 = |U| \ge \frac{1}{4}|A| + 2 = \frac{3}{4} + 2$. Property \ref{crown-prop-f-touches-only-a}: considering $\pi_{G \to R}(w^i) \subseteq \{ f'_i, f'_{i+1}, f'_{i+2}\}$ we conclude that no part-path outside of $A$ is intersected. Property \ref{crown-prop-f-path-covering}: by the clipping, $T$ intersects $\pth_{f'_i}$ and $\pth_{f'_{i+2}}$ in a single node. Furthermore, $T$ intersects $\pth_{f'_{i+1}}$ in at most $D$ consecutive nodes due to the replacement; the property is satisfied.

  This concludes the proof because we constructed valid disjoint crowns containing (in union) exactly $x+1 \ge \frac{3}{10}|V(f)|$ parts.
\end{proof}

We finalize the proof of the main result of this section by applying \Cref{crowns-on-a-path} to all $\{ q_i \}_{i \in J}$ and concatenating the collections of crowns constructed this way (they are clearly disjoint since they are supported on disjoint paths of $R$). We establish a disjoint collection of crowns $\mathcal{C} = \{ (T_i, A_i, U_i) \}_{i \in I}$ with $\sum_{i \in I} |A_i| \ge \frac{1}{10}k \cdot \frac{3}{10} = \frac{3}{100} k$. This completes the proof of \Cref{crowns-low-degree}. \hfill  $\blacksquare$

\subsubsection{Finalizing the crown construction}

Combining the high-degree case (\Cref{crowns-high-degree}) and the low-degree case (\Cref{crowns-low-degree}) with crown merging (\Cref{crown-merging}) we conclude the crown construction with the following result.

\begin{lemma}\label{full-crown-construction}
  For every set of $k$ minimal part-paths $\{\pth_i\}_{i=1}^k$, there always exists a crown $(T, A, U)$ with respect to $\{ \pth_i \}_{i \in A}$, where $U \subseteq A \subseteq [k]$ and $|U| \ge \frac{1}{280}k$.
\end{lemma}
\begin{proof}
  Let $R$ be the contraction graph of $\{\pth_i\}_{i=1}^k$. By applying \Cref{crowns-high-degree} and \Cref{crowns-low-degree} to $R$, we can find a set of disjoint crowns $\mc{C} = \{(T_i, A_i, U_i)\}_{i \in I}$ where $\sum_{i \in I} |A_i| \ge \min(\frac{1}{70} k, \frac{3}{100} k) = \frac{1}{70} k$. Merging the crowns via \Cref{crown-merging} we construct a single crown $(T_*, A_*, U_*)$ of $\{\pth_i\}_{i \in A_*}$ satisfying $|A_*| = \sum_{i \in I} |A_i| \ge \frac{1}{70} k$. We deduce that $|U_*| \ge \frac{1}{4}|A_*| + 2 \ge \frac{1}{280} k$ by Property \ref{crown-prop-u-vs-a} of crowns.
\end{proof}

\subsection{Converting crowns into relaxed disjointness gadget}\label{sec:crown-to-disjointness}

In this section we prove the following result.
\begin{restatable}{lemma}{fullmststructure}\label{full-mst-structure}
  Given a set of $k$ connectable pairs $\{ (s_i, t_i) \}_{i=1}^k$, there always exists a subset $U \subseteq [k]$ of size $|U| \ge \frac{1}{1400}k$ and a relaxed disjointness gadget $(P, T)$ with endpoints $\{(s_i, t_i)\}_{i \in U}$.
\end{restatable}
We prove this result by the following lemma, which converts a large crown (in terms of number of parts belonging to a it) into a large relaxed disjointness gadget (\Cref{def:relaxed-disjointness}).

\begin{lemma}\label{crown-to-disjointness}
  Let $(T, A, U)$ be a crown with respect to $\{\pth_i\}_{i=1}^k$ with endpoints $\{(s_i, t_i)\}_{i=1}^k$. There exists a subset $U' \subseteq U$ of size $|U'| \ge \frac{1}{5} |U|$ and a relaxed disjointness gadget $(P, T)$ with endpoints $\{(s_i, t_i)\}_{i \in U'}$.
\end{lemma}
\begin{proof}
  We build a directed ``interference graph'' $I$: the vertices correspond to parts $U$ and the outgoing edges of a part $i \in U$ are defined by the following.

  For every $i \in U$ let $\ps_i$ be the shortest walk in $G$ from $s_i$ to the closest vertex in $V(T)$. If $\ps_i$ touches another part-path, we clip-off anything beyond the first touch, i.e., $\mr{ps}_i \gets \clip(\mr{ps}_i, \{s_i\}, \bigcup_{j \in U, j \neq i} V(\pth_j))$. In this case, let $j$ be the part index $\ps_i$ touches. We add the directed edge $i \to j$ to the interference graph $I$. Finally, we ``associate'' the walk $\ps_i$ to part $i$ (regardless of whether it touches another part or not).

  We repeat the exact same steps for $t_i$: let $\pt_i$ be the shortest walk from $t_i$ to $V(T)$; we clip-off a suffix and add an edge $i \to j'$, if needed; finally, associate $\pt_i$ with part $i$. Exactly two walks are associated with each part.

  Let $I'$ be the undirected version of $I$ (directed edges $i \to j$ are transformed to undirected edges $\{i, j\}$). Since the out-degree of $I$ is at most two, the average degree of $I'$ is at most $4$. Then, by Turan's theorem, there exists an independent set $U' \subseteq U$ in $I'$ with $|U'| \ge \frac{1}{1+d}k \ge \frac{1}{5} k$ where $d \le 4$ is the average degree. We will call the parts $U \setminus U'$ (i.e., outside of the independent set) ``sacrificial''.

  We initialize $T'$ with by adding all part-paths of sacrificial parts to $T$, i.e., $T' \gets T \cup \bigcup_{j \in U \setminus U'} E(\pth_j)$. Note that this $T'$ is connected since the crown ensures $T$ touches each part-path. We will later adjust $T'$ and we maintain its connectivity.

  As warm-up, suppose that for all $i \in U'$ (called ``useful'' parts) all walks associated with $i$ do not intersect $\pth_i$ (except at the starting endpoint of the walk). In this case, for each $i \in U'$ (in arbitrary order) we do the following. Consider each of the two walks $f$ associated with $i$ (in arbitrary order). We reassign $T' \gets T' \cup f$, ensuring $\{s_i, t_i\} \ni f_0 \in V(T')$ (i.e., $T'$ touches an endpoint of $\pth_i$). We also maintain the connectivity of $T'$ because either $f_{|f|} \in V(T)$ (i.e., no clipping occurred, $f$ connects to $T$, in which case the claim is obvious) or it connects to another part-path $\pth_j$, which necessarily must be sacrificial (i.e., $j \not \in U'$) and therefore added to $T$ since $i \in U'$ and $U'$ is an independent set. After the process is finished for all $i \in U'$ we have that $(\{\pth_i\}_{i \in U'}, T')$ is a relaxed disjointness gadget: we already argued that connectivity of $T'$ is maintained; $s_i, t_i \in V(T')$ since the two walks associated with $i \in U'$ contain $s_i$ and $t_i$ as endpoints and always get added to $T'$. Lastly, we argue about covering $V(T') \cap V(\pth_i)$: we have that $T' = T \cup \bigcup_{j \in U \setminus U'} E(\pth_j) \cup \bigcup_{j \in J} f_j$ where $\{ f_j \}_{j \in J}$ is a collection of walks associated with parts in $U'$ (i.e., useful parts). For some $i \in U'$ we have:
  \begin{align*}
    V(T') \cap V(\pth_i) & = \left(V(T) \cap V(\pth_i)\right) \cup \left(\bigcup_{j \in U \setminus U'} V(\pth_j) \cap V(\pth_i)\right) \cup \left(\bigcup_{j \in J} V(p_j) \cap V(\pth_i)\right) \\
                         & \subseteq \left(V(T) \cap V(\pth_i)\right) \cup \emptyset \cup \{ s_i, t_i \}.
  \end{align*}
  Here we used that $\pth_i, \pth_j$ are vertex-disjoint for $i, j \in U, i \neq j$; walks $\{ f_j \}_{j \in J}$ do not intersect useful part-paths except in endpoints $\{s_i, t_i\}$; $V(T) \cap V(\pth_i)$ can be covered by one path of length at most $D$ (from Property~\ref{crown-prop-f-path-covering} of crowns). In conclusion, $V(T') \cap V(\pth_i)$ can be covered by $\{s_i\}, \{t_i\}$ and one sub-path of length at most $D$; $(T', U')$ is a disjointness gadget. This completes the warm-up.
  
  We now consider general $f$, i.e., we allow $f$ associated with $i$ to intersect $\pth_i$. To achieve this, we will need to iteratively adjust the part-paths. We initialize $\pth'_i \gets \pth_i,\ \forall i \in U'$ and adjust $\{ \pth'_i \}_{i \in U'}$ as needed. We will maintain the following invariant: for all vertices $v \in V$ that, at any point, are in the set $v \in \bigcup_{j \in U'} V(\pth'_j)$ and $v \not \in \bigcup_{j \in U'} V(\pth_j)$ (i.e., were not in the same set for the original definition of part-paths), it will hold that $v \in V(T')$.

  For each $i \in U'$ (in arbitrary order) we do the following. Consider each walk $f$ associated with part $i$ (in arbitrary order). We first examine whether internal nodes of $f$ intersect $V(T')$ and apply $f \gets \clip(f, \{ f_0 \}, V(T'))$ if this is the case.

  Let $t$ be the largest step such that $f_t \in \pth'_i$ (e.g., $t = 0$ in the warm-up scenario). We replace the prefix/suffix of $\pth'_i$ between $f_0$ and $f_t$ with $f_{[0, t]}$. Note that this does not change the endpoints of $\pth'_i$. We finally update $T' \gets T' \cup f$. This maintains the invariant that all vertices added to $\pth'_i$ are in $V(T')$ since $V(f_{[0, t]}) \subseteq V(f) \subseteq V(T')$.
 
  We argue this maintains connectivity of $T'$. We discuss a few possibilities. If internal nodes of $f$ intersected $V(T')$ the connectivity of $T' \cup f$ is clear. If it did not intersect $T'$ then we have the same possibilities as in the warm-up: either $f_{|f|} \in T$, or the endpoint could belong to a sacrificed part (as in the warm-up). In all of these cases $f$ is connected to $T'$ and the connectivity of $T'$ is maintained.

  After the process is finished for all $i \in U'$, we argue that $(\{\pth'_i\}_{i \in U'}, T')$ is a relaxed disjointness gadget. The connectivity of $T'$ is satisfied as previously argued. We have that $s_i, t_i \in V(T')$ since the two walks associated with $i$ contain $s_i$ and $t_i$ as endpoints (this property is maintained after potential clipping) and get added to $T'$.

  Finally, we argue that $V(T') \cap V(\pth'_i)$ for $i \in U'$ can be covered by at most three sub-paths of $\pth'_i$ of length at most $D$. By construction, $\pth'_i = f^i_{[0,t]} \circ (\pth_i)_{[a, b]} \circ q^i_{[t', 0]}$ ($a, b, t, t'$ depend on $i$ but we will drop this for notational simplicity) where $f^i, q^i$ are the two walks associated with part $i$ (possibly empty), $q^i_{[t', 0]}$ represents the walk in reverse order $(q^i_{t'}, q^i_{t'-1}, \ldots, q^i_{0} )$ and $\circ$ concatenates walks with matching endpoints. The construction stipulates that $T' = T \cup \bigcup_{j \in U \setminus U'} E(\pth_j) \cup \bigcup_{j \in J} f_j$ where $\{ f_j \}_{j \in J}$ is a collection of walks associated with parts in $U'$ (i.e., useful parts).

  By construction, $V(f^i_{[0,t]}) \subseteq V(T')$. Since $f^i$ is a shortest path in $G$ and the diameter of $G$ is at most $D$ we have that $|f^i| \le D$, hence $|f^i_{[0,t]}| \le D$. In other words, the intersection $V(f^i_{[0,t]}) \cap V(T') = V(f^i_{[0,t]})$ can be covered by a sub-path of length at most $D$. The same holds for $q^i$.

  The intersection of $V(T') \cap V((\pth_i)_{[a,b]}) = V(T) \cap V((\pth_i)_{[a,b]})$ for $i \in U'$: this follows from $T' = T \cup \bigcup_{j \in U \setminus U'} E(\pth_j) \cup \bigcup_{j \in J} f_j$, the original part-paths being vertex-disjoint, and $\{ f_j \}_{j \in J}$ are (clipped versions of) walks that do not intersect $(\pth_i)_{[a,b]}$ of any useful part $i$ (they do not intersect non-associated useful parts and the restriction to $[a,b]$ avoids associated walks). By Property~\ref{crown-prop-f-path-covering}, the intersection of $V(T) \cap V(\pth_i)$ can be covered by one sub-path of length at most $D$. Thereby, restricting to $(\pth_i)_{[a, b]}$, we conclude that $V(T') \cap V((\pth_i)_{[a,b]})$ can be covered by one sub-path of length at most $D$.

  In conclusion, remembering that $\pth'_i = f^i_{[0,t]} \circ (\pth_i)_{[a, b]} \circ q^i_{[t', 0]}$, we conclude that $V(\pth'_i) \cap V(T')$ can be covered by at most three sub-paths of length at most $D$ (one for each factor in the concatenation). This concludes the proof.
\end{proof}

We now complete the proof of \Cref{full-mst-structure}.
\begin{proof}
  Let $\{\pth_i\}_{i=1}^k$ be minimal part-paths with endpoints $\{(s_i, t_i)\}_{i=1}^k$.
  Using \Cref{full-crown-construction}, we find a crown $(T_*, A_*, U_*)$ of $\{\pth_i\}_{i \in A_*}$ with $|U_*| \ge \frac{1}{280}k$. Furthermore, \Cref{crown-to-disjointness} guarantees the existence of a relaxed disjoint gadget $(P, T)$ with endpoints $\{(s_i, t_i)\}_{i \in U}$ satisfying $|U| \ge \frac{1}{5}|U_*| = \frac{1}{1400} k$.
\end{proof}

\subsection{Finalizing the disjointness gadget}

In this section we use the developed tools to extend every moving cut to a disjointness gadget.

\begin{lemma}\label{lemma:movingcut->gadget}
If a set of $k$ connectable pairs $S$ in $G$ admit a moving cut with capacity $\tO(k)$ and distance $\beta \ge 9D$, then $G$ contains an $\tOmega(\beta)$-disjointness gadget.
\end{lemma}
\begin{proof}
  Let $S = \{(s_i, t_i)\}_{i=1}^k$ be a set of connectable pairs, admitting an $(\alpha,\beta)$-moving which we denote by $\ell$. We can assume WLOG that the capacity is at most $k / 30001$ due to \Cref{betas-bounds} which allows us to scale-down both the capacity and distance of $\ell$ by $\poly \log n$ factors (note that this lemma is insensitive to such factors).
  
  By \Cref{full-mst-structure}, there exists a relaxed disjointness gadget $(P,T)$ with respect to some subset of these pairs, $S\subseteq \{(s_i, t_i)\}_{i=1}^k$, of size $|S|\geq \frac{1}{1400} k$.
  We will further focus on the subset $S'\subseteq S$ of pairs $(s_i,t_i)$ whose path $p_i:s_i\rightsquigarrow t_i$ in $P$ is made up entirely of edges $e$ with $\ell$-length precisely one (smallest possible).
  The capacity of $\ell$ is at most $\sum_{e}(\ell_e - 1) < k/30000 < |S|/2$.
  Consequently, since the paths $p_i$ are disjoint and each edge $e$ with $\ell_e>1$ contributes at least one to $\sum_{e}(\ell_e - 1)$, the set $S'\setminus S$ contains at most $k/30000$ pairs, and so $|S'|\geq |S|-k/30000 \geq k/3000$. Note that the distance and capacity (upper bound) of $\ell$ when move from $S$ to its subset $S'$.
  
  Now, consider a source-sink pair $(s_i,t_i)\in S'$, connected by path $p_i\in P$.
  By the definition of a relaxed disjointness gadget, the intersection of $p_i$ with $T$ is covered by a set $\Phi_i$ of at most 3 sub-paths of $p_i$, of hop-length at most $D$. Note that, by definition of relaxed disjointness gadget, sub-paths in $\Phi_i$ always cover endpoints $\{s_i, t_i\}$.
  By our choice of $S'$, the sub-paths in $\Phi_i$, which have hop-length at most $D$, also have $\ell$-length at most $D$ (since each of their edges' $\ell$-length is precisely one).
  Now, consider the set $\Psi_i$ of (one or two) sub-paths of $p_i$ obtained by removing the sub-paths $\Phi_i$ from $p_i$. For each such sub-path $p\in \Psi_i$, denote by $s(p)$ and $t(p)$ the first and last node of $p$. Then, since $\ell$ has distance $\beta$ for $S$, we have in particular that
  \begin{align*}
    \sum_{p\in \Psi_i} d_{\ell}(s(p),t(p)) + \sum_{p\in \Phi_i} |p| \geq d_{\ell}(s_i,t_i) \geq \beta.
  \end{align*}

  Now, since $|p|\leq D$ for each $p\in \Phi_i$, and therefore $\sum_{p\in \Phi_i} |p|\leq 3D$, and since $\Psi_i$ contains no more than two sub-paths of $p_i$, there must be some sub-path $p'_i\in \Psi_i$ of $p_i$, with endpoints $s'_i:=s(p'_i)$ and $t'_i:=t(p'_i)$ for which 
  \begin{equation}\label{eqn:relaxed-to-strict}
    d_{\ell}(s(p'_i),t(p'_i))\geq (\beta-3D)/2. 
  \end{equation}

  By our choice of $S'$, all edges in $p_i:s_i\rightsquigarrow t_i$ for $(s_i,t_i)\in S'$ have $\ell$-length of one. 
  Consequently, by \Cref{eqn:relaxed-to-strict}, all such paths $p'_i$ have hop-length at least $|p'_i|=\ell(p'_i)\geq (\beta-3D)/2\geq 3$, since $\beta\geq 9D$.
  We now show that this implies that $(P',T)$ induces a (strict) disjointness gadget for the endpoints of these sub-paths $P':=\{p'_i\}_i$.
  Indeed, the subgraph $H := T\cup \bigcup_{i: (s_i,t_i)\in S'}(p_i\setminus p'_i)$ is connected, and only intersects the paths $p'_i$ (of length $|p'_i|\geq 3$) at their first and last vertices. 
  Consequently, for $T'$ a spanning tree of $H$, the pair $(P',T')$ is a disjointness gadget with endpoints $S'$.

  Denote by $S' = \{(s(p'_i),t(p'_i))\}_{i \in I'}$ where $I'$ is a set of indices. Here we note that we have shown, via construction, that $\dist_\ell(s(p'_{\pmb{i}}), t(p'_{\pmb{i}})) \ge (\beta - 3D) / 2 = \Omega(\beta)$ for all $i \in I'$. However, the moving cut requires a lower bound on $\dist_\ell(s(p'_{\pmb{i}}), t(p'_{\pmb{j}}))$ for all $i, j \in I'$, which might not be true in general. We invoke a structural lemma \cite[Lemma 2.5]{haeupler2020network}: given $k$ source-sink pairs $\{(s_i, t_i)\}_{i=1}^k$ in a (general) metric space with $\dist(s_i, t_i) \ge \beta$, one can find a subset $I \subseteq [k]$ of size $|I| \ge k/9$ such that $\dist(s_i, t_j) \ge \beta / O(\log k)$ for all $i, j \in I$. Applying this result, we find a subset $\tilde{I} \subseteq I'$ of size $|\tI| \ge |I'| / 10 \ge k / 30000$ where $\dist_\ell(s(p'_i), t(p'_i)) \ge \tOmega(\beta)$ for all $i, j \in \tI$. Therefore, $\ell$ has capacity at most $k / 30001 < |\tI|$ and distance at most $\tOmega(\beta)$ with respect to $\{(s_i, t_i)\}_{i \in \tI}$. Hence $(\{p'_i\}_{i \in \tI}, T', \ell)$ is a $\tOmega(\beta)$-disjointness gadget.
\end{proof}

Armed with \Cref{lemma:movingcut->gadget}, we may finally prove \Cref{beta_D>=beta_C}, restated below.
\betadandbetac*
\begin{proof}
	If $\movingcut{G} < 9D$, then trivially $\wcsub{G} \leq \movingcut{G} < 9D$. Therefore, since $\movingcut{G} \geq D$ always (a length assignment of one to all edges yields a moving cut of capacity $0$ and distance $D$ for any pair of maximally-distant nodes), we have that 
	\begin{align*}
	\wcsub{G}+D = \Theta(D) = \Theta(\movingcut{G}).
	\end{align*}
	If, conversely, $\movingcut{G}\geq 9D$, then by \Cref{lemma:movingcut->gadget}, there exists an $\tOmega(\movingcut{G})$-disjointness gadget, and therefore $\wcsub{G} = \tOmega(\movingcut{G}) \geq \tOmega(D)$. On the other hand, since we trivially have that $\wcsub{G} \leq \movingcut{G}$, we find that 
	\begin{align*}
	\wcsub{G}+D & = \tilde{\Theta}(\wcsub{G}) = \tilde{\Theta}(\movingcut{G}).
	\qedhere
	\end{align*}
\end{proof}


\section{Algorithmic Shortcut Construction: Matching the Lower Bound}\label{sec:shortcut-construction}

In this section we give our shortcut construction for the supported CONGEST model which efficiently constructs shortcuts of quality $\tilde{O}(Q)$ in $\tilde{O}(Q)$ rounds, where $Q=\shortcut{G}$ is the best possible shortcut quality.
\shortcutsupported*

By \Cref{thm:shorcutsimplyMSTetal}, numerous problems, including MST, approximate min-cut, approximate shortest path problems, and verification problems can therefore be solved in $\tilde{O}(\shortcut{G})$ rounds when the topology is known. On the other hand, from \Cref{sec:lower-bound} we know that all these problems are harder than the subgraph connectivity problem, which requires $\tilde{\Omega}(\shortcut{G})$ rounds on any network $G$, by \Cref{conn>=Q(G)}. These matching bounds together give \Cref{thm:universally-Optimal}.

\subsection{Constructing shortcuts for pairs}\label{sec:pairwise-via-oblivious-routing}

In order to construct shortcuts for parts, we will rely on the ability to construct shortcuts for pairs. 
In particular, we will require such shortcut construction for pairs of nodes which are oblivious of each other.
To this end, we rely on the existence of hop-constrained oblivious routing to construct such shortcuts between pairs of nodes, yielding the following lemma.

\begin{lemma}\label{pairwise-oracle}
  There exists a supported CONGEST algorithm that, given $k$ disjoint pairs of nodes $S = \{ (s_i, t_i) \}_i$ (each $s_i$ knows $t_i$ and vice versa, but not other pairs) in a graph $G$, constructs an $\tilde{O}(Q)$-quality shortcut for $S$ in $\tilde{O}(Q)$ rounds, where $Q=\pairshortcut{G}$. 
\end{lemma}
\begin{proof}
Knowing the topology, for $h=2^1,2^2,\dots,2^{\lceil \log_2 n\rceil}$, all nodes internally compute (the same) $h$-hop oblivious routing with $\tilde{O}(1)$ hop stretch and $\tilde{O}(1)$ congestion approximation for all values $h$, denoted by $\{R^{h}_{s,t}\}_{s,t,h}$. Sampling these distributions then gives paths $p^h_{s,t}\sim R^{h}_{s,t}$ for each such triple $(s,t,h)$.
	By \Cref{shortcuts-from-routing} and \Cref{thmGeneralRouting}, for $h\in [Q,2Q)$, the set of paths $p^h_{s_i,t_i}$ form shortcuts for $S$ of quality $q:=O(Q\cdot \log^7n) = O(h\cdot \log^7n)$.

	During the shortcut construction stage, we appeal to the random-delay-based routing protocol which implies that $\routing{S} = \tilde{\Theta}(\pairshortcut{S})$ (i.e.,  \Cref{LMR}). 
	In particular, for $h=2^1,2^2,\dots,2^{\lceil \log_2 n\rceil}$, for each pair $(s_i,t_i)$, we send a message  between $s_i$ and $t_i$ via the path $p_i$, starting at a uniformly randomly chosen time in $[q]$, and send this message during $O(q\cdot \log n)$ rounds.
	We let this message contain the identifiers of $s_i$ and $t_i$, and so intermediary nodes, which all know $p^h_{s_i,t_i}$, can forward this message along the path.
	By standard random-delay arguments \cite{leighton1994packet}, if the paths $p^h_{s_i,t_i}$ are shortcuts of quality $q$, then all pairs $(s_i,t_i)$ will have both of their messages delivered w.h.p.
	To verify whether or not all sinks $t_i$ receive their message, 	after these $O(q\cdot \log n)$ rounds, all sinks $t_i$ flood a single-bit message through the system, indicating whether any sink has not received its designated message from $s_i$. This step takes $O(D)$ rounds.
	Now, since by \Cref{shortcuts-from-routing} we know that for $h\in [Q,2Q]$, the sampled shortcuts are $q=\tilde{O}(h)$-quality shortcuts w.h.p., and since $D\leq Q = \pairshortcut{G}$, this algorithm terminates successfully after $$\sum_{k=1}^{\log_2Q+1} (\tilde{O}(1)\cdot 2^i + D) = \tilde{O}(Q)$$ 
	rounds w.h.p.
	\end{proof}

	We now turn to leveraging this $\tilde{O}(\pairshortcut{G})$-quality shortcuts construction for pairs to obtain $\tilde{O}(\shortcut{G})$-quality shortcuts for parts. 

\subsection{Lifting shortcuts for pairs to shortcuts for parts}\label{sec:pair-to-general-lift}

In \Cref{sec:pairwise-partwise-algorithms} we describe how to use the shortcut construction for pairs to construct general low-congestion shortcuts, giving the following lemma.

\begin{restatable}{lemma}{constructivePairsToParts}\label{constructive-pairs-to-parts}
  Suppose there exists an algorithm that for input connectable pairs $\calS = \{ (s_i, t_i) \}_{i=1}^k$ (each $s_i$ knows $t_i$ and vice versa, but not other pairs) outputs a shortcut for $\calS$ with quality $\tilde{O}(Q)$ in $T$ rounds, where $Q=\pairshortcut{G}$. Then there exists a randomized CONGEST algorithm that on input disjoint and connected parts $\calP = (P_1, \ldots, P_k)$ (with each node only knowing the $i$ for which $v\in P_i$, if any), constructs a shortcut for $\calP$ of quality $\tilde{O}(Q)$ in $\tilde{O}(T)$ rounds w.h.p.
\end{restatable}

Finally, invoking \Cref{constructive-pairs-to-parts} with \Cref{pairwise-oracle} as its pairwise shortcut algorithm, we obtain a supported CONGEST algorithm which constructs $\tilde{O}(Q)$-quality shortcuts for any disjoint connected parts in $\tilde{O}(Q)$, for $Q=\shortcut{G}$. That is, we have proved \Cref{partwise-shortcut-supported}.

        
\section{Conclusions and Open Questions}\label{sec:conclusion}

In this work we give the first non-trivial universally-optimal distributed algorithms for a number of optimization and verification problems in the supported CONGEST model of communication. 
In particular, we show that the low-congestion framework yields such universally-optimal algorithms. This framework has since shown great promise as a basis for a unified description of the common communication bottlenecks underlying many distributed network optimization problems.

This work suggests a number of natural follow-up questions, of which we mention a few here.

\medskip\noindent\textbf{Universally-optimal CONGEST algorithms.} Our algorithmic results require efficient computation of shortcuts, in time $\tilde{O}(\shortcut{G})$. Using the recent oblivious routing result of \cite{ghaffari2020hop}, we showed how to achieve this efficiently in the supported CONGEST model. Can such shortcuts be computed in the CONGEST model, i.e., without pre-processing or knowledge of $G$? A positive answer to this question would yield universally-optimal CONGEST algorithms for the problems tackled by the low-congestion shortcut framework.

\medskip\noindent\textbf{Universally-optimal algorithms for more problems.} Our work proves that the low-congestion framework, which yields existentially-optimal $\tilde{O}(D+\sqrt{n})$-time algorithms for numerous distributed optimization problems (see \cite{ghaffari2016mst,dassarma2012distributed}), in fact yields universally-optimal algorithms for these problems.
However, some fundamental problems remain for which shortcut quality serves as a universal lower bound (by our work), yet shortcut-based algorithms are not currently known. One such problem is \emph{exact} min-cut, for which an existentially-optimal algorithm was recently given in \cite{dory2020distributed}. 
Another such problem is (better) approximate SSSP, for which the best approximation guarantee of a shortcut-time algorithm is polynomial \cite{haeupler2018faster}, while the best existentially-optimal algorithm yields a $(1+\eps)$ approximation \cite{becker2017near}.
Our work motivates the study of shortcut-based (and shortcut-quality time) algorithms for these and other problems, as such algorithms would be universally optimal, by our work.

\medskip\noindent\textbf{Universally-Optimal Round and Message Complexity.}
Another well-studied complexity measure of message-passing algorithms is their \emph{message complexity}, i.e., the number of messages they send during their execution. 
Recently \citet{pandurangan2017time} showed that existentially-optimal time and message complexities are achievable simultaneously, resolving a longstanding open problem.
This was then shown to be achievable deterministically, by \citet{elkin2017simple}, and then shown to be achievable within the shortcut framework, in \cite{haeupler2018round}. We note, however, that  \cite{haeupler2018round} relied specifically on tree-restricted shortcuts. Can one remove this restriction? 
A positive resolution to this question would yield algorithms with \emph{universally-optimal} time complexity, and optimal message complexity.

\medskip\noindent\textbf{Other universal barriers to distributed computation.} 
Our work shows that for a wide family of problems for which $\tilde{\Theta}(D+\sqrt{n})$ serves as a tight existential bound, shortcut quality serves as a tight universal bound. 
Can similar tight universal bounds be proven for problems outside this ``complexity class''?


\section*{Acknowledgements} This research was supported in part by NSF grants CCF-1910588, NSF CCF-1812919, CCF-1814603, CCF-1618280, CCF-1527110, NSF CAREER award CCF-1750808, ONR award N000141912550, Swiss National Foundation (project grant 200021-184735), a Sloan Research Fellowship, and a gift from Cisco Research.

\appendix
\section*{\LARGE Appendices}
\section{Sub-Diameter Bounds in the Supported CONGEST}\label{sec:sub-diam}

Distributed computing often studies ``local'' problems like the maximal independent set which can be solved in sub-diameter time (those are mostly studied in the LOCAL model with unlimited message sizes~\cite{linial1987distributive, Linial92}, as opposed to CONGEST).
This is in stark contrast to \emph{global} problems in distributed verification and optimization (e.g., like the MST or connected subgraph verification problem), where the diameter term has often been regarded as a strict necessity. 
Indeed, as we show in  \Cref{sec:optimality-formal-definitions}, the graph diameter $D$ is a universal (supported) CONGEST lower bound for these problems.

In this section, we show that the diameter is not a universal lower bound in the supported CONGEST model for other well-known ``local'' variants of problems which we study. 
On the other hand, we will leverage our universally-optimal bounds for the problems studied in the paper body to obtain universally-optimal bounds for these variants, too.
For concreteness and brevity, we only look at the MST and spanning connected subgraph as the representative problems. We now define the local variants of these problems. 

\medskip
\noindent \textbf{MST-edge.} This variant of MST has the same input as the variant discussed in \Cref{sec:models-and-problems}. However, the nodes do not have to know the cumulative weight of the MST. Rather, each node $v$ must determine for each of its incident edges whether it is in the MST, and all outputs must be globally consistent.

\smallskip   
\noindent \textbf{Conn-veto.} This variant accepts the same input as the connected subgraph verification problem discussed in \Cref{sec:models-and-problems}. Furthermore, each node continues to output a value which denotes whether subgraph $H$ is connected and spanning or not. However, this output does not need to be the same between all nodes, but rather if $H$ is connected and spanning all nodes need to answer YES, while in the opposite case at least one node needs to answer NO.

For these variants, the diameter $D$ is \emph{not} the right parameter in the supported CONGEST, as suggested by prior work. For example, consider a tree: in this graph MST-edge is trivially solvable in zero rounds. Instead, we show that this diameter term should be replaced by the \emph{maximum diameter of any biconnected component} in $G$, which we denote by $D_{bcc}(G)$.

We note that the (standard, non-veto) spanning connected subgraph verification problem does not reduce to this MST-edge problem, as the former admits a lower bound of $D$, which, as noted before, does not hold for MST-edge. We denote via $T_{\text{MST-edge}}(G)$ the worst-case running time of any always-correct algorithm $\calA$ for MST-edge in supported CONGEST, i.e.~(in the notation of \Cref{sec:optimality-formal-definitions}),we define $T_{\text{MST-edge}}(G) := \max_I T_{\calA}(G, I)$, and define $T_{\text{conn-veto}}(G)$ analogously.

\begin{observation}\label{MST-edge-to-conn-veto}
  For any graph $G$ it holds that $T_{\text{MST-edge}}(G) \geq T_{\text{conn-veto}}(G)$.
\end{observation}
\begin{proof}
	We rely on the standard reduction of \cite{dassarma2012distributed} from MST to spanning subgraph  connectivity.
	Briefly, edges of $H$ have weight zero and edges of $G\setminus H$ have infinite weight. 
	The subgraph $H$ is a spanning connected subgraph in $G$ iff the MST has weight zero, or put otherwise, if all nodes only have zero-weight edges in the MST. 
	Therefore, any node with a non-zero-weight edge in the MST can output NO, thus solving the veto version of the spanning connected subgraph problem. 
\end{proof}

We now turn to characterizing the time to solve both problems $\Pi\in \{\text{MST-edge},\text{conn-veto}\}$, which as we see, is tightly related to the graph's biconnected components, as well as its worst-case sub-network, and these biconnected components' shortcut quality.

\subsection{The role of biconnectivity}

We start with a simple lower bound, showing that the maximum diameter of any biconnected component, denoted by $D_{bcc}(G)$, serves as a lower bound for these problems.

\begin{lemma}\label{bcc-D-LB}
  For a problem $\Pi\in \{\text{MST-edge},\text{conn-veto}\}$ in any graph $G$ it holds that 
  $$T_{\Pi}(G) \geq \frac{D_{bcc}(G) - 1}{2}.$$
\end{lemma}
\begin{proof}
	We start by proving this lower bound for $\Pi = conn-veto$.
	Let $\mathcal{B}$ be a a biconnected component of highest diameter, $D(G[\mathcal{B}])=D_{bcc}(G)$, and let $u$ and $v$ be two nodes in $\mathcal{B}$ of highest distance in $G':=G[\mathcal{B}]$. 
	By Menger's theorem, $G'$ contains two node-disjoint paths $p:u\rightsquigarrow v$ and $p':v\rightsquigarrow u$. 
	Together, these subgraphs define a cycle $\mathcal{C}$ in $G'$. 
	Consider the spanning connected subgraph problem in a (random) subgraph $H$ which contains a spanning forest of $G$ after contracting the cycle $\mathcal{C}$. In addition, $H$ contains all edges of $\mathcal{C}$, except for possibly the first edge along both paths $p$ and $p'$, which we denote by $e$ and $e'$ respectively. The subgraph $H$ either contains neither of the edges $e$ and $e'$, or exactly one of these, with each of these three options happening with probability $\nicefrac{1}{3}$. 
	The subgraph $H$ is clearly spanning, and it is disconnected iff both $e$ and $e'$ do not belong to $H$.
	Therefore, for the algorithm to succeed on $H$ with probability greater than $\frac{2}{3}$, some node must determine whether or not both edges $e,e'$ do not belong to $H$, which requires at least  $\frac{d_G(u,v)-1}{2}$ rounds of communication, as this information originates at endpoints of $e$ and $e'$, which are at distance at least $d_G(u,v)-1$ from each other.
	From this, we obtain that 
	\begin{align}
	T_{conn-veto}(G) \geq \frac{d_G(u,v)-1}{2}.\label{distance-in-cycle}
	\end{align}
	On the other hand, by the biconnectivity of $\mathcal{B}$, the shortest (necessarily simple) path between $u$ and $v$ is completely contained in $\mathcal{B}$, since no simple  $u\rightsquigarrow v$ path contains nodes outside of $\mathcal{B}$.
	Thus, we have
	\begin{align}
	d_G(u,v)\geq D_{bcc}(G).
	\label{distance-in-cycle-as-in-graph}
	\end{align}	
	Combining \eqref{distance-in-cycle} and \eqref{distance-in-cycle-as-in-graph}, the lemma follows for $\Pi = conn-veto$. The lower bound for $\Pi = \text{MST-edge}$ then follows from the lower bound on conn-veto together with \Cref{MST-edge-to-conn-veto}.
\end{proof}

Next, we show that the time to solve both problems studied in this section is equal to the time to solve the problem on their biconnected components.
More formally, denoting by $\mathcal{B}(G)$ the set of biconnected components of $G$, we show the following.
\begin{lemma}\label{tight-for-bcc}
	The time to solve either problem $\Pi\in \{\text{MST-edge}, \text{conn-veto}\}$ in supported CONGEST satisfies
	\[
	T_{\Pi}(G) = \max_{B\in \mathcal{B}(G)} T_{\Pi}(G[B]).
	\]
\end{lemma}
\begin{proof}
	First, we prove the upper bounds $T_{\Pi}(G) \leq \max_{B\in \mathcal{B}(G)} T_{\Pi}(G[B])$ for $\Pi\in \{\text{MST-edge}, \text{conn-veto}\}$. For these, we essentially note that running an algorithm with optimal optimal running time for each biconnected component in parallel yields an algorithm for $G$ whose running time is precisely $\max_{B\in \mathcal{B}(G)} T_{\Pi}(G[B])$. In more detail, for the veto verification problem, a subgraph $H$ is a spanning connected subgraph of $G$ iff it contains a spanning connected subgraph of each biconnected component. Therefore, after solving conn-veto in each biconnected component, letting each node output NO iff it outputs NO in any of these components results in a correct solution to conn-veto in $G$. For MST-edge, a spanning tree of $G$ must contain a spanning tree in each biconnected component, and a minimum spanning tree must contain an MST in each such biconnected component (the converse would result in a strictly better MST for $G$ obtained by replacing the spanning tree in $B$ by an MST of $B$).
	
	We now turn to proving the lower bounds of the form $T_{\Pi}(G) \geq \max_{B\in \mathcal{B}(G)} T_{\Pi}(G[B])$ for both problems. Here, again we rely on the solution to an input supported on $G$ also yielding a solution the same input restricted to $G[B]$, for any $B\in \mathcal{B}(G)$. 
	Now, let $\calA$ be a correct algorithm solving any input supported on $G$ in time $T$. Then, this implies a $T$-round correct algorithm $\calA'$ for any input supported on $G[B]$: 
	on input not supported on $G[B]$, this algorithm simulates some arbitrary correct algorithm.
	On input to problem $\Pi$ supported on $G[B]$, Algorithm $\calA'$ simulates Algorithm $\calA$ with trivial data on the edges of of $G$ outside $B$. That is, for $\Pi = \text{MST-edge}$, these edges all have weight zero, and for $\Pi = \text{conn-veto}$, all these edges exist. 
	It is not hard to see that an optimal solution to the extended input on $G$ yields a solution for the input on $G[B]$. 
	On the other hand, as $B$ is biconnected, each node $v$ not in $B$ has a single node $b_v\in B$ such that for any $b'\in B$, any path $v\rightsquigarrow b'$ passes through $b$ (the alternative would imply two disjoint paths from $v$ to some node $b'\in B$, and would therefore have $v$ be in $B$). 
	Consequently, each node $b$ can simulate the messages sent from nodes $v$ with $b=b_v$, since $b$ knows all initial messages of such $v$, as well as all messages sent through $b$.
	Algorithm $\calA'$ therefore solves any input on $G[B]$ in time $T = T_\Pi(G)$. 
	Since this algorithm is also correct, we conclude that $T_\Pi(G) = T \geq T_{\Pi}(G[B])$.
\end{proof}

A simple corollary of lemmas \ref{bcc-D-LB} and \ref{tight-for-bcc} is that the classic $\tilde{\Theta}(D+\sqrt{n})$ should really be restated as $\tilde{\Theta}(D_{bcc}(G)+\sqrt{n})$. 
As in previous sections, we now turn to tightening this further, providing universally-optimal supported CONGEST algorithms for both problems.

In order to obtain our universally-optimal tight bounds,  we next note that the time to solve the $\Omega(D)$-time variants of the above problems trivially differ by their (potentially sub-diameter-time) variants studied in this section by at most an additive diameter term.
\begin{observation}\label{loca-to-global-plus-diam}
	For any graph $G$, 
	\[
	T_{conn}(G)\leq T_{conn-veto}(G) + 2D(G).
	\]
	\[
	T_{MST}(G)\leq T_{MST-edge}(G) + 2D(G).
	\]
\end{observation}
\begin{proof}
	For both problems, first solve the more ``local'' problem and then
	perform a BFS from some agreed-upon node, and use the obtained tree to either sum the weights of edges of the MST, or sum the number of NO answers by performing a parallel convergecast in the tree, and then broadcast the result to all nodes.
\end{proof}

The above then allows us to obtain optimal supported CONGEST algorithms for both problems of this section.
\begin{theorem}
	The time to solve any problem $\Pi\in \{\text{MST-edge}, \text{conn-veto}\}$ in the supported CONGEST model is
	\begin{align*}
	T_{\Pi}(G)
	& = \tilde{\Theta}\left(\max_{B\in \mathcal{B}(G)} \shortcut{G[B]} + D_{bcc}(G)\right).
	\end{align*}
\end{theorem}
\begin{proof}
	By \Cref{tight-for-bcc}, we know that for both problems $\Pi\in \{\text{MST-edge}, \text{conn-veto}\}$, we have that
	$T_{\Pi}(G) = \max_{B\in \mathcal{B}(G)} T_{\Pi}(G[B])$. But by \Cref{loca-to-global-plus-diam}, for any $B\in \mathcal{B}[G]$, we have
	\[
	T_{conn-veto}(G[B]) \geq T_{conn}(G[B]) - 2\cdot D(G[B])
	\]
	On the other hand, by \Cref{bcc-D-LB}, we have that
	\[
	T_{conn-veto}(G[B]) \geq D(G[B])/2 - 1
	\]
	Combining both inequalities, we get
	\[
	7\cdot T_{conn-veto}(G[B])\geq T_{conn}(G[B]) + D(G[B]) - 6.
	\]
	Consequently, since disjointness gadgets are biconnected and are therefore contained in some $G[B]$ for $B\in \mathcal{B}[G]$, we have that
	\begin{align*}
	T_{conn-veto}(G) & = \max_{B\in \mathcal{B}(G)} T_{conn-veto}(G[B]) \\
	& = \Omega\left(\max_{B\in \mathcal{B}(G)} T_{conn}(G[B]) + D(G[B])\right) \\
	& = \tOmega\left(\max_{B\in \mathcal{B}(G)} \shortcut{G[B]} + D_{bcc}(G)\right),
	\end{align*}
	where the last lower bound follows from \Cref{conn>=Q(G)}.
	The same lower bound holds for MST-edge.
		
	Finally, since the low-congestion shortcut framework allows to solve both of the problems $\Pi\in \{\text{MST-edge}, \text{conn-veto}\}$ in each biconnected component $B$ in time
	$\tilde{O}\left(\shortcut{G[B]} + D(G[B])\right)$ by \Cref{partwise-shortcut-supported} and \Cref{thm:shorcutsimplyMSTetal}, and since $T_{\Pi}(G) \leq \max_{B\in \mathcal{B}(G)} T_{\Pi}(G[B])$, we obtain the matching upper bound, and the theorem follows.
\end{proof}


\section{Deferred Proofs of \Cref{sec:prelims}: Preliminaries}\label{sec:prelims-appendix}

In this section we present proof of lemmas deferred from \Cref{sec:prelims}, restated for ease of reference. 

\subsection{Moving Cuts}
We start with a simple proof of the relationships between some key graph parameters which we study.

\relationships*
\begin{proof}
  Let $\mathcal{C}$ denote the set of connectable pairs in $G$, and let $S^*_C := \arg\max_{S\in \calC} \comm{S}$. By definition, we have that $\comm{S^*_C} = \comm{G}$ and by lemmas~\ref{moving-cuts-communication} and \ref{LMR} we also have
  \begin{align*}
    \comm{S^*_C} & = \comm{G} = \tilde{\Theta}(\routing{S^*_C})  = \tilde{\Theta}(\routing{G}) \\
    & = \tilde{\Theta}(\pairshortcut{S^*_C}) = \tilde{\Theta}(\pairshortcut{G}) .
  \end{align*}

  We now have to show that $\movingcut{G}$ also fits in the above chain of (polylog-approximate) equalities. Let $S^*_{MC} := \arg\max_{S\in \calC} \movingcut{S}$. By lemma \ref{moving-cuts-communication} it holds that there exists $S' \subseteq S^*_C$ such that
  \begin{align*}
    \comm{G} & = \comm{S^*_{MC}} = \tilde{\Theta}(\movingcut{S'}) \\ 
             & \le \tilde{\Theta}(\movingcut{S^*_{MC}}) = \tilde{\Theta}(\movingcut{G}) .
  \end{align*}

  On the other hand, by definition we have $\movingcut{G} = \movingcut{S^*_{MC}}$ and $\comm{S^*_{MC}} \le \comm{G}$. Furthermore, \cite[Lemma 1.7]{haeupler2020network} shows that for any set of pairs $S$ it holds that $\movingcut{S} \le \comm{S}$. Hence, combining, we get the final result.
  \begin{align*}
    \comm{G} & \ge \comm{S^*_{MC}} \ge \movingcut{S^*_{MC}} \\
             & = \movingcut{G} \ge \tilde{\Omega}(\comm{G}) . \qedhere
  \end{align*}
\end{proof}

We now state and prove two auxiliary lemmas needed to prove \Cref{lemma:moving-cut-simulation}. The first of these is a simple argument concerning scaling of the capacity and distance of moving cuts.

\begin{restatable}{fact}{betasbounds}\label{betas-bounds}
  A moving cut with capacity $\gamma$ and distance $\beta$ can be transformed into a moving cut with capacity $\gamma/c$ and distance $\beta / (1 + c)$, for any $c \ge 1$.
\end{restatable}
\begin{proof}
  Let $\ell$ be a moving cut of capacity $\gamma$ and distance $\beta$ between $S = \{ (s_i, t_i) \}_{i=1}^k$.
  Consider a new moving cut $\hat{\ell}_e := 1+\lfloor \frac{\ell_e -1}{c}\rfloor$. This assignment $\hat{\ell}$ has capacity at most $\sum_{e} \hat{\ell}_e - 1 \leq \sum_e \frac{\ell_e - 1}{c} < \frac{\gamma}{c}$. So $\hat{\ell}_e$ indeed has capacity $\gamma / c$.
  
  We analyze the distance $\hat{\beta}$ of $\hat{\ell}$. Consider a source $s_i$ and sink $t_j$ in $S$. For any path $p:s_i\rightsquigarrow t_j$ of hop-length $|p|\geq \frac{1}{1+c}\cdot \beta$, clearly we have that this path's $\hat{\ell}$-length is at least $\hat{\ell}(p)\geq \sum_{e\in p} 1 = |p| \geq \frac{\beta}{1+c}$. For any path of hop-length $|p|\leq \frac{\beta}{1+c}$, we have that this path's $\hat{\ell}$-length is at least 
  $$\hat{\ell}(p) = \sum_{e\in p} \hat{\ell}_e = \sum_{e\in p} 1+ \left\lfloor\frac{\ell_e - 1}{c}\right\rfloor \geq \sum_{e\in p} \frac{\ell_e - 1}{c} \geq \frac{\ell(p) - |p|}{c} \geq \frac{\beta - \frac{1}{1+c}\cdot \beta}{c} = \frac{1}{1+c}\cdot \beta,$$
  where the last inequality relied on $|p|\leq \frac{1}{1+c}\cdot \beta$ and on $\ell$ having distance $\beta$.
  As this argument holds for any source $s_i$ and sink $t_j$, the distance of $\hat{\ell}$ is at least $\dist_{\hat{\ell}}(\{s_i\}_{i=1}^k, \{t_j\}_{j=1}^k) \geq \frac{\beta}{1+c}$. We conclude that $\hat{\ell}$ has indeed capacity $\gamma / c$ and distance $\beta / (1 + c)$ for $S$.
\end{proof}

The following lemma allows us to transform an efficient distributed algorithm into an efficient communication complexity solution. The proof follows the arguments (implicitly) contained in \cite{dassarma2012distributed,haeupler2020network}.

\begin{lemma}\label{pure-moving-cut-simulation}
  Let $\calA$ be a distributed computation of $f : \{0,1\}^k \times \{0,1\}^k \to \{0, 1\}$ between $\{s_i\}_{i=1}^k$ and $\{t_i\}_{i=1}^k$ with running time at most $T$. If there exists a moving cut $\ell$ between $\{s_i\}_{i=1}^k$ and $\{t_i\}_{i=1}^k$ of capacity $\gamma$ and distance at least $2T$ then there is an $O(\gamma \log n)$-communication complexity protocol $\calC$ for $f$. The error probability of the protocol is the same as the distributed algorithm's error probability.
\end{lemma}
\begin{proof}
  We use a simulation argument to convert $\calA$ into an efficient protocol $\calC$, using the moving cut $\ell$ as a guide. Naturally, both Alice and Bob know $G$, $\ell$ and $\calA$, but not each others' private inputs $\{x_i\}_{i \in [k]}$ and $\{y_i\}_{i \in [k]}$.
  
  We introduce some proof-specific notation. For node $v \in V$, we denote the shortest $\ell$-distance from $v$ to any node in $\{s_i\}_i$ by $d_v := \dist_\ell(\{s_i\}_{i \in [k]}, v)$. 
  Moreover, we denote by $L_{< t} := \{ v \in V \mid d_v < t \}$ the nodes within distance $t$ of some node in  $\{s_i\}_i$, and similarly we let $L_{>  t} := \{v\in V \mid d_v > t\}$. We label the message-sending rounds of $\calA$ with $\{1, 2, \ldots, T\}$, and define a \emph{timestep} $t \in \{0, 1, \ldots, T\}$ to be the moments of ``inactivity'' (during which the message history of each node is constant) after round $t$. So, timestep $0$ is the moments before any message is sent; timestep $t > 0$ is the moments between round $t$ and $t+1$; timestep $T$ is the moments after $\calA$ completed.
  
  The simulation proceeds as follows. At timestep $t \in \{0, 1, \ldots, T \}$ Alice will simulate all nodes in $L_{< 2T-t}$ and Bob will simulate all nodes in $L_{>  t}$. 
  By simulating a node $v$, we mean that the player knows all messages received by $v$ up to that timestep and the private input of $v$. 
  (Note that only $\{s_i\}_i$ and $\{t_i\}_i$ have private inputs.) 
  Such a simulation at timestep $t=T$ would imply in particular then Alice and Bob will know the answer $f(x, y)$, since $\calA$ completed by timestep $t=T$, and so all nodes know the answer, including $\{s_i\}_{i \in [k]} \subseteq L_{< 2T-T}$ and $\{t_i\}_{i\in [k]}\subseteq L_{> T}$ (by the moving cut's distance). 
  We now construct a communication complexity protocol $\calC$ by sequentially extending a valid simulation at timestep $t-1$ to one at timestep $t\leq T$ (i.e., we simulate the exchange of round-$t$ messages of $\calA$). We make this argument only for Alice, while  Bob's side follows analogously.	
  Initially, Alice can simulate $L_{< 2T - 0}$ without any messages from Bob, since no nodes have received any messages and she knows the private inputs $\{x_i\}_{i \in [k]}$ of $\{s_i\}_{i \in [k]} \subseteq L_{< 2T - 0}$ and $\{t_i\}_{i \in [k]} \cap L_{< 2T-0} = \emptyset$ (again, by the moving cut's distance); vice versa for Bob. 	
  
  Next, let $v$ be a node simulated by Alice in timestep $t\in (0,T]$.
  Since $v\in L_{< 2T-t} \subseteq L_{< 2T-(t-1)}$, Alice simulated $v$ in timestep $t-1$, and thus at timestep $t$, Alice knows the messages sent to $v$ in rounds $1,2,\dots,t-1$. So, all she needs to learn to simulate $v$ at timestep $t$ are the messages sent to $v$ at round $t$ in $\calA$. 
  If such a message comes from a neighbor $w$ simulated by Alice in the prior timestep, $w \in L_{< 2T-(t-1)}$, then Alice has all the information needed to construct the message herself. Conversely, if $w \not \in L_{< 2T-(t-1)}$, then we argue that $w$ was simulated by Bob in timestep $t-1$. This is because $w \not \in L_{< 2T-(t-1)}$ implies $d_w \ge 2T-t+1 \ge T+1 > t > t-1$ (due to $t \le T$), i.e., $w \in L_{>  t-1}$. When this is the case, we say the edge $\{v, w\}$ is \emph{active}. For each active edge of round $t$, Bob sends Alice the $O(\log n)$-bit message sent via $\calA$ from $w$ to $v$, by serializing it over $O(\log n)$ one-bit communication complexity rounds and we append those rounds to $\calC$.
  So, denoting by $\calC_t$ be the set of messages that Bob needs to send to Alice to transition from timestep $t-1$ to timestep $t$ (in some pre-determined order), then $\calC$ is the concatenation of $\calC_1, \calC_1, \ldots, \calC_T$. By construction, after receiving the messages of $\calC_t$, Alice can simulate timestep $t$ of $\calA$. Finally, it is guaranteed that Alice learns the answer by timestep $t = T$ when $\calA$ completes.
  
  It remains to prove that the communication complexity protocol $\calC$ uses only $O(\gamma \log n)$ bits.
  For each edge $e=\{v,w\}\in E$, Bob sends at most $O(\log n)$ bits in $\calC$ for every timestep $e$ is active.
  By definition, this edge is active when $v \in L_{< 2T - t}$ and $w \not \in L_{<2T-(t-1)}$. Equivalently, $\{v,w\}$ is active when $d_v < 2T - t$ and $d_w \ge 2T - t + 1$, or, $2T + 1 - d_w \le t < 2T - d_v$. Therefore, $e$ is active during $2T - d_v - (2T + 1 - d_w) = d_w - d_v - 1$ rounds. We have that $d_w \le d_v + \ell_e$ since $d_{(\cdot)}$ are distances with respect to $\ell$. Therefore, an edge is active at most $d_w - d_v - 1 \le \ell_e - 1$ rounds. Adding over all edges $e \in E$ we have that the total round-complexity of $\calC$ is $O(\log n)\cdot \sum_{e \in E} (\ell_e - 1) = O(\gamma \log n)$. We note that this accounts only for the information sent from Bob to Alice. The analysis for the messages sent from Alice to Bob is completely analogous and requires the same amount of rounds. Interlacing these rounds increases the communication complexity by a multiplicative factor of $2$, leading to a final $O(\gamma \log n)$-round communication complexity protocol $\calC$ for $f$, as claimed.  
\end{proof}
Note that the simulation argument even applies to distributed algorithms with shared randomness between the nodes.

We now show that \Cref{pure-moving-cut-simulation} implies a lower bound on distributed disjointness computation in the face of a moving cut.

\simulation*
\begin{proof}
  Let $C > 0$ be a sufficiently large constant. Suppose there is a $T$-time distributed computation $\calA$ for $\disj$ with $T \le \frac{\beta}{2 (1 + C \log n)}$. We can transform the moving cut to have capacity $\frac{k}{C\log n}$ and dilation $\frac{\beta}{1 + C \log n}$ via \Cref{betas-bounds}. Applying the simulation \Cref{pure-moving-cut-simulation} we construct a $O(k) / C$-communication complexity solution for $\disj$. However, $\disj$ is known to have communication complexity $\Omega(k)$, even for constant-probability error protocols \cite{razborov1990distributional}, leading to a contradiction for a sufficiently large constant $C$. Thus we have that $T \ge \Omega(\beta / \log n)$ and we are done.
\end{proof}

Here we explain why \cite{haeupler2020network} implies that moving cuts for pairs of nodes provide a characterization (up to polylog factors) for the time to solve multiple unicasts for these pairs.
\pairwise*
\begin{proof}
  In the terminology of \cite{haeupler2020network}, each edge has some capacity $c_e$ denoting the number of bits transmittable across an edge in one round (in our case $c_e= O(\log n)$). Moreover, each of the $k$ pairs $(s_i,t_i)$ has some demand $d_i$, corresponding to a message size of $d_i$ bits to send from $s_i$ to $t_i$. In our setting, $d_i=1$ for all pairs in the set $S$.

  Let $\beta := \routing{S}$ be the optimal completion time of a routing protocol that completes the multiple unicast task on $S$. The result of \cite[Lemma 2.3]{haeupler2020network} shows that there exists a subset of pairs $S' \subseteq S$ such that $\movingcut{S'} \ge \tilde{\Omega}(\beta)$. Furtermore, \cite[Lemma 1.7]{haeupler2020network} shows the existence of a moving cut with capacity strictly less than $k$ and distance $\movingcut{S'}$ for $S$ implies $\comm{S} \geq \movingcut{S'}$.

  Furthermore, since a multiple unicast on a superset $S$ is also a solution on a subset $S'$ we have that $\comm{S} \ge \comm{S'}$. And finally, $\routing{S} \ge \comm{S}$ by definition since routing algorithms solve the multiple unicast task.
  
  Combining these bounds together, we obtain the following chain of inequalities, implying the lemma.
  \begin{align*}
    \routing{S} & \ge \comm{S} \ge \comm{S'} \ge \movingcut{S'} \ge \tilde{\Omega}(\routing{S}).
  \end{align*}
\end{proof}

\subsection{Oblivious routing}

Here we prove that hop constrained routing schemes imply shortcuts of bounded congestion and dilation.
\shortcutsrouting*
\begin{proof}
	The fact that the sampled paths are shortcuts is immediate. 
	Moreover, the dilation of the paths is at most $\beta\cdot h = O(\beta\cdot Q)$ by construction.
	For the congestion bound, we focus on paths given by a $Q$-quality shortcut of $S$, which is guaranteed to exist for $S$ by definition of $Q=\pairshortcut{G}$.
	Now, we let $R^*$ be a routing scheme which contains for each $(s_i,t_i)\in S$ the distribution $R_{s_i,t_i}$ consisting of the single path $p_i:s_i\rightsquigarrow t_i$ in the shortcuts, and for any pair $(s,t)\not\in S$ picks an arbitrary $s\rightsquigarrow t$ path.
	Consider the demands $\calD_{s,t}:=\mathds{1}[(s,t)\in S)]$. Now, since $R^*$ was obtained from shortcuts of dilation $Q$, we have that the fractional congestion of $R^*$ w.r.t.~$\calD$ is at most $$\cong(\calD, R^*) = \max_{e \in E} \E_{p \sim R^*_{s,t}} \left[ \calD_{s, t} \cdot \mathds{1}[e \in p ] \right ] = \max_{e \in E} \sum_{i} \mathds{1}[e\in p_i] \leq Q.$$
	Therefore, we obtain that $\opt\at{h}(\calD) \leq \cong(\calD,R^*) \leq Q$.
	Hence, since $R$ is an $h$-hop oblivious routing with congestion approximation  $\tilde{O}(1)$, we have that this routing's congestion w.r.t.~$\calD$ is $\cong(\calD,R^*) \leq \alpha \cdot  \opt\at{h}(\calD) \leq \alpha \cdot Q$. 
	Combined with the definition of $\cong(\calD,R)$, this implies that $\max_{e \in E} \allowbreak \E_{p \sim R_{s,t}} \allowbreak \left[ \calD_{s, t} \cdot \mathds{1}[e \in p ] \right ]\leq \alpha \cdot Q$. 
	Consequently, sampling a path $p_i\sim R_{s_i,t_i}$ for each pair $(s_i,t_i)$ results in any given edge belonging to at most $\alpha\cdot Q$ sampled paths in expectation, and at most $O(\log n)\cdot \alpha\cdot Q$ w.h.p., by standard chernoff bounds. Taking union bound over all edges proves that the sampled paths' congestion is $O(\log n)\cdot \alpha \cdot Q$ w.h.p.
\end{proof}


\section{Deferred Proofs of \Cref{sec:lower-bound}: Lower Bounds}\label{sec:lower-bound-appendix}

In this section we provide the proofs of lemmas deferred from \Cref{sec:lower-bound}, restated below for ease of reference.
\subsection{$\beta$-disjointness gadgets as lower bounds certificates}

Here we prove that $\beta$-disjoint gadgets are indeed witnesses of a lower bound on the time to solve subgraph spanning connectivity. We first show that the diameter is a universal lower bound for the spanning connected subgraph, and hence (via known reductions), all other problems we study in this paper.

\begin{lemma}\label{lemma:D-lower-bound-for-connectivity}
  Any always-correct spanning connected subgraph algorithm in the supported CONGEST model (hence, also in CONGEST) requires $\Omega(D)$ rounds in every diameter $D$ graph $G$.
\end{lemma}
\begin{proof}
  If $n\leq 2$ or $D=O(1)$, the observation is trivial. Suppose therefore that $n\geq 3$ and $D\geq 3$.
  Let $s$ and $t$ be two vertices of maximum hop distance in $G$ and let $e$ be the edge in $p$ incident to $s$. In particular, $d_G(s,t)=D$.
  Let $T$ be a spanning tree of $G / p$ (with edges of $p$ contracted).
  We now consider two different inputs $H_0, H_1$ to the spanning connected subgraph problem:
  \begin{align*}
  1[e' \in H_b] = \begin{cases}
    b & e' = e, \\
    1 & e' \in E(T), \\
    1 & E(p) \setminus \{e\}, \\
    0 & \text{otherwise}
  \end{cases}
  \end{align*}
  It is easy to see that the $H_0$ is disconnected and $H_1$ is connected (and spanning). Since node $t$ needs to know the answer, it needs to be able to differentiate between $H_0$ and $H_1$. However, only nodes incident to $s$ have a different input, and those two nodes are at distance at least $D - 1$ from $t$. Therefore, $t$ cannot receive any differentiating information before round $D - 1$, completing our lower bound.
\end{proof}

With this prerequisite in place, we now prove the following result.

\witnesslb*

\begin{proof}
  By \Cref{lemma:D-lower-bound-for-connectivity}, we have $T_{conn}(G) \ge \Omega(D)$.  
  We prove that $T_{conn}(G) \ge \tilde{\Omega}(\wcsub{G})$. 
  
  Let $(P,T,\ell)$ be a $\beta$-disjointness gadget.
  Let the endpoints of the paths $\pth_i\in P$ be $\{(s_i, t_i)\}_{i=1}^k$. We denote the edges of $\pth_i$ by $E(\pth_i)$, and their strictly internal edges (i.e., edges not incident to $s_i$ or $t_i$) by $E^{\circ}(\pth_i)$ .

  Suppose we want to perform distributed computation of the disjointness function between $\{s_i\}_i$ and $\{t_i\}_i$. For each $i\in [k]$, $s_i$ and $t_i$ know a private bit $x_i$ and $y_i$, respectively. We define a subgraph $H$ which is spanning and connected if and only if $\disj(x,y)=1$. For a set $F \subseteq E$ let $G / F$ denote $G$ with edges $F$ contracted. Let $H_1$ be some spanning tree of $G / ( T \cup \bigcup_{i \in [k]} E(\pth_i) )$, which must exist since $G$ is connected. Let $H_2 = T \cup \bigcup_{i \in [k]} E^\circ(\pth_i)$. We construct $H_3$: for every path $p_i$, the edge of $p_i$ incident to $s_i$ (resp., $t_i$) belongs to $H_3$ iff $x_i=0$ (resp., $y_i = 0$). Finally, let $H := H_1 \cup H_2 \cup H_3$. Each node can precompute $H_1$ and $H_2$ without extra information (they know $G$ and $\{\pth_i\}_i$), and each node can use its private info to compute which of its incident edges is in $H_3$.

  

  It is easy to check that $H$ is connected if and only if $\disj(x,y)=1$ (i.e., $x_i\cdot y_i = 0$ for all $i\in [k]$). Thus, a $T$-round algorithm for spanning subgraph connectivity immediately yields a $T$-round distributed computation of $\disj$ between $\{s_i\}_{i=1}^k$ and $\{t_i\}_{i=1}^{k}$. 
  By \Cref{lemma:moving-cut-simulation}, combining this algorithm and the moving cut of the $\beta$-disjointness gadget implies $T = \Omega(\beta)$.
\end{proof}

\subsection{Relating shortcuts for pairs and for parts }\label{sec:multicast-to-unicast}

In this section we show an equivalence (up to polylog factors) between shortcuts for connectable pairs of nodes and for parts. We first present the existential result. We first show the result that is only existential, and then proceed to give an algorithmic version of it.

\shortcutsrelation*

The non-trivial part of the above lemma is the following upper bound on $Q(G)$. 
\begin{lemma}\label{Q<=Q2}
	For any graph $G$, 
	$$\shortcut{G} = \tO(\pairshortcut{G}).$$
\end{lemma}
\begin{proof}
	Let $S_1,S_2,\dots,S_k$ be a partitioning of $V$ into subsets of nodes inducing connected subgraph $G[S_i]$. We show that shortcuts for this partition can be obtained by combining any (pair-wise) shortcuts for some polylogarithmic number of pairs of vertices in $\mathcal{C}$, implying the lemma.

	For our proof, we make use of heavy-light tree decompositions \cite{sleator1983data}, which decompose a tree on $n$ nodes into sub-paths, with each root-to-leaf path in the tree intersecting $O(\log n)$ of these sub-paths.
	For each $i\in [k]$, we consider some heavy-light decomposition of $T_i$. Note that since the parts are disjoint, these trees are disjoint, and consequently, so are the obtained sub-pats of the decomposition.
	Letting $q:=\pairshortcut{G}$, we first show that the partition with parts $V(p)$ for all sub-paths $p$ of some tree decomposition of some $T_i$ admits shortcuts of quality $O(q\cdot \log n)$.
	
	Indeed, shortcuts of quality $O(q\cdot \log n)$ for parts $P$ whose parts are disjoint paths, as above, can be defined recursively, as follows. For a path $p\in P$, we denote by $s(p),m(p)$ and $t(p)$ the first, median, and last node on this path $p$. We note that using sub-paths of these $p$, we find that the pairs $\{(s(p),m(p)),(m(p),t(p)) \mid p\in P\}$ are connectable.
	Thus, these pairs admit shortcuts of quality $q$. 
	We use these shortcuts (and more edges which we will choose shortly) as the shortcuts for parts $V(p)$.
	We then consider recursively the two sub-paths of all such paths $p$, one starting at $s(p)$ and ending at $m(p)$, and the other starting at $m(p)$ and ending at $t(p)$, noting that these sub-paths for all paths $p\in P$ are also connectable. 
	We recursively compute shortcuts for all such sub-paths and add the shortcuts for each sub-path of $p$ to the shortcuts for each part $V(p)$.
	We show that the union of shortcuts for all sub-paths at each of the $O(\log n)$ levels of recursion are shortcuts for the parts $P$ above.
	First, as these shortcuts for $P$ are the union of $O(\log n)$ shortcuts for pairs of quality $q$, then the congestion of these shortcuts is trivially $O(q\cdot \log n)$. In order to bound dilation, we prove by induction that for each node $v$ in one of these sub-paths of length $|p|\leq 2^{k}$, the shortcuts of the sub-paths of $p$ constructed this way contain a path of length at most $q\cdot k$ from $v$ to $s(p)$, and likewise to $t(p)$. 
	Indeed, consider the subpath $p'$ containing $v$ and $m(p)$, which has length at most $|p|/2\leq 2^{k-1}$. Then, since $m(p)$ is either $s(p')$ or $t(p')$, there exits a path of length at most $q\cdot (k-1)$ from $v$ to $t(p)$. Concatenating this path with the shortcut path of length (at most) $q$ from $m(p)$ to $s(p)$, and likewise to $t(p)$, yields the desired result.
	Therefore, any two nodes in a path $p$ can be connected by a walk of length at most $q\cdot 2\log_2n$ contained in the shortcuts of its recursively-defined sub-paths. Thus $O(q\cdot \log n)$ upper bounds the dilation of these shortcuts for $p$. We conclude that any partition $P$ whose parts are disjoint paths admits shortcuts of quality $O(q\cdot \log n)$.
	
	So far, we have shown that the partition given by parts induced by the sub-paths of tree decompositions of spanning trees $T_i$ of parts $S_i$ admit shortcuts of quality $O(q\cdot \log n)$. We now show that the union of these sub-paths' shortcuts, where each part has as shortcuts the union of the shortcuts of its' tree decomposition's sub-paths, are shortcuts of quality $O(q\cdot \log^2n)$ for the partition $S_1,S_2,\dots,S_k$. First, the congestion of these shortcuts is $O(q\cdot \log^2 n)$, since each edge belongs to at most $O(q\cdot \log n) = O(q\cdot \log^2n)$ sub-paths' shortcuts, and therefore it belongs to the shortcuts of at most $O(q\cdot \log^2n)$ distinct parts.
	Second, since any two nodes $u,v$ in $S_i$ have that their $T_i$-path consists of at most $O(\log n)$ sub-paths, using the shortcuts for these $O(\log n)$ sub-paths induce a path of length  $O(q\cdot \log n) \cdot O(\log n) = O(q\cdot \log^2n)$ between $u$ and $v$.
	Thus, each partition $S_1,S_2,\dots,S_k$ admits $O(q\cdot \log^2 n) = \tilde{O}(q) = \tilde{O}(\pairshortcut{G})$-quality shortcuts, as claimed.
\end{proof}

The following simple lemma yields the complementary ``opposite'' inequality.

\begin{lemma}\label{Q2<=Q}
	For any graph $G$,
	$$\pairshortcut{G} \leq \shortcut{G}.$$
\end{lemma}
\begin{proof}
	Let $S\in \mathcal{C}$ be a set of connectable pairs, and let $P$ be a set of disjoint paths connecting these pairs. Now, we consider these paths (or more precisely, their vertices) as parts of a partition with one more part for the remaining nodes, $V\setminus \bigcup_{p\in P}v(p)$. Then, by definition of $\shortcut{G}$, there exist shortcuts for these parts of quality at most $\shortcut{G}$. Consequently, these shortcuts contain each edge at most $\shortcut{G}$ times. Moreover, for each $i$, these shortcuts contain a path of length at most $\shortcut{G}$ connecting each pair of nodes $(s_i,t_i)$, as these nodes are in a common part. Put otherwise, the best shortcuts for the parts defined above imply the existence of shortcuts of quality no worse than $\shortcut{G}$ for $S$.
\end{proof}

\Cref{shortcut-relation} follows directly from lemmas \ref{Q<=Q2} and \ref{Q2<=Q}.


\section{Deferred Proofs of \Cref{sec:constructing-disjointness}: Crown Constructions}\label{sec:constructing-disjointness-appendix}

In this section we prove the simple crown-merging lemma, whereby disjoint crowns can be merged to form a single crown on the union of the crowns' parts.

\crownmerging*
\begin{proof}
	Consider the shortest path connecting two vertices on part-paths belonging to different crowns. More precisely, let $V(A_i) := \bigcup_{x \in A_i} V(\pth_x)$. Then we consider $\min_{i, j\in I; i \neq j} \dist_G(V(A_i), V(A_j))$. Let $q$ be this shortest path and suppose its endpoints are $u \in V(\pth_x)$ and $v \in V(\pth_y)$, with $x \in A_i$ and $y \in A_j$. By the minimality of $q$ we know that the internal vertices of $q$ do not belong to any crown, where we say that a vertex $u$ \emph{belongs} to a crown $(T, A, U)$ if $u \in V(\pth_i)$ and $i \in A$. We now merge the crowns $(T_i, A_i, U_i)$ and $(T_j, A_j, U_j)$ into a single crown $(T_*, A_*, U_*)$. We remove crowns $i$ and $j$ from the list of crowns $I$, add $*$ to the list, and recurse until there is only one crown in the list. The following construction ensures that $A_* = A_i \sqcup A_j$ and that the new crown $*$ is disjoint from all other crowns in the list $I$. The final crown $(T_*, A_*, U_*)$ will clearly have the property that $A_* = \bigsqcup_{i \in I} A_i$ and will be defined with respect to $\{\pth_i\}_{i \in A_*}$.
	
	To help with properly defining the merge, we define auxiliary sets $(T'_i, A'_i, U'_i) \gets (T_i, A_i, U_i)$ and $(T'_j, A'_j, U'_j) \gets (T_j, A_j, U_j)$. If the endpoint of $q$, belonging to crown $i$, is in a useful part $x$ (i.e., $q \in \pth_x$ and $x \in U_i$), then we declare the part $x$ that $q$ intersects sacrificial and add $\pth_x$ to $T'_i$. More precisely, if $x \in U_i$ then we define $U'_i \gets U_i \setminus \{x\}$ and $T'_i \gets T_i \cup E(\pth_x)$. Note that $T'_i$ is still connected since $V(T'_i)$ intersects $V(\pth_x)$. We analogously do the same for crown $j$: if the endpoint of $q$ belonging to crown $j$ is in a useful part, we declare that part sacrificial and add its path to $T'_j$.
	
	Define the crown $(T_{*}, A_{*}, U_{*}) \gets (T'_1 \cup T'_2 \cup E(q), A_1 \cup A_2, U'_1 \cup U'_2)$, add it to the set of remove $(T_i, A_i, U_i)$ and $(T_j, A_j, U_j)$. We check it is a proper crown. $T_{*}$ is connected since $T'_1$ and $T'_2$ are connected and $q$ is a path connecting them. Property \ref{crown-prop-u-vs-a}: $|U_*| \ge |U_i| + |U_j| - 2 \ge \frac{1}{4}|A_i| + 2 + \frac{1}{4}|A_j| + 2 - 2 = \frac{1}{4}|A_*| + 2$. Properties \ref{crown-prop-f-touches-only-a} and \ref{crown-prop-f-path-covering} are satisfied because they were satisfied in $(T'_i, A'_i, U'_i)$ and $(T'_j, A'_j, U'_j)$ and (by minimality of $q$) no internal nodes of $q$ intersect a part-path belonging to a crown.
\end{proof}


\section{Deferred Proofs of \Cref{sec:shortcut-construction}: From Pairwise to Partwise Shortcuts}\label{sec:pairwise-partwise-algorithms}

\Cref{shortcut-relation} shows the existence of structures implying an equivalence (up to polylog factors) between $\shortcut{G}$ and $\pairshortcut{G}$. 
We now show that similar equivalence (again, up to polylogs) holds for the distributed (CONGEST) complexity of computing approximately-optimal shortcuts for parts and pairs. 
In particular, we show how to obtain an algorithm for shortcut construction for parts using $\tilde{O}(1)$ applications of an algorithm for shortcut construction for pairs.

\constructivePairsToParts*

\begin{proof}
  \textbf{Terminology and notation.} Given a ``node function'' $f : V \to \mathcal{X}$, we say that a CONGEST algorithm ``distributedly knows'' $f$ if every node $v \in V$ knows $f(v)$. Similarly, an ``edge function'' $f : E \to \mathcal{X}$ is ``distributedly known'' if each node $v \in V$ knows $f(e)$ for all its incident edges. Furthermore, consider an edge function $f : E \to \{\to, \leftarrow, \bot\}$. We denote by $G(f)$ the directed graph $(V, E')$ where the edge $e \in E$ is either not in $E'$ if $f(e) = \bot$, and is otherwise directed one way or the other. We prove a sequence of intermediary results.

  \textbf{Subclaim 1:} Suppose that $f : E \to \{\to, \leftarrow, \bot\}$ is such that $G(f)$ is precisely a union of (maximal) vertex-disjoint directed paths $P_1 \cup P_2 \ldots \cup P_q$. We argue there is a distributed $\tilde{O}(T)$-round algorithm that takes the distributed knowledge of $f$ as input and (distributedly) learns for each node $v$ (i) the unique path ID of the maximal path $P_i \ni v$ it is contained in, the size of the path $|P_i|$ and the first node on the path (ii) the depth (distance of $v$ from the start of the $P_i$), (iii) a shortcut with respect to $(P_1, P_2, \ldots, P_q)$ of quality $\pairshortcut{G}$. Moreover, given a distributed function $g : V \to [n^{O(1)}]$, each node can also learn (iv) the sum $\sum_{u \in R(v)} g(u)$ where $R(v)$ is the set of nodes $u$ reachable from $v$ via the directed edges of $G(f)$ (i.e., they are on the same path $P_I$, but ``below'' $v$).
  
  \textbf{Proof of subclaim 1:} The algorithm starts with initially singleton clusters and then grows the clusters for $O(\log n)$ steps. In the first step, each node is its own cluster. In each subsequent step we assume each node $v$ knows the data (i)--(iv) for each cluster (i.e., subpath) that $v$ is in. We show how to merge neighboring clusters while maintaining the validity of this data. First, each cluster flips an independent random heads/tails coin (the entire cluster flips a common coin without communication via shared randomness). Intuitively, tails will merge into heads and only if ``heads subpath'' has a directed edge towards the ``tails subpath''. The merging is fairly simple. For each head subpath $A$ (i.e., a cluster that flipped heads) and each tail subpath $B$ where there is a directed ``connecting'' edge of $G(f)$ from $A$ to $B$, we do the following. We pass the unique path ID of the head subpath via the connecting edge into the tail subpath $B$ and disseminate it throughout the nodes of $B$ via the computed shortcut in $\tilde{O}(T)$ rounds. Similarly, we can disseminate the sizes $|A|, |B|$ and the starting node of the path $A$ to correctly all values stipulated by (i). For (ii), we add the value of $|A|$ to the depths of each node in the $B$ via the shortcut of $B$. For (iii), we add the shortcut between the start of $A$ to the start of $B$ (note that pairs of connecting clusters do this operating in parallel, but all of these paths are connectable, hence the operation is valid; we argue in the next paragraph the quality suffices). Finally, for (iv), we simply calculate the sum of $g$'s in $B$ via the shortcut, propagate this sum $\sigma$ to $A$ and add $\sigma$ to the sum of $g$'s in each node of $A$.

  Clearly, each pair of neighboring subpaths (i.e., neighboring clusters) will merge with probability at least $1/4$ in each step. Furthermore, on each maximal path $P_i$ each cluster will have at least one other neighboring cluster unless the entire path is a single cluster. Therefore, after $O(\log n)$ merging steps the entire path is a single cluster with high probability via standard arguments. Finally, we argue about the shortcut quality of the clusters. First, since the shortcut is constructed by $O(\log n)$ calls to the pairwise construction (\Cref{pairwise-oracle}), the congestion is clearly $\tilde{O}(\pairshortcut{G})$. We argue the same dilation bound by consider an arbitrary node $v$ and arguing that after $i$ merging steps it can reach the start of its own cluster in $i \cdot \pairshortcut{G}$ hops: in each merging step the subpath $v$ belongs to either is not merged (we do not do anything), it is a head subpath being merged (we do not do anything), or it is a tail subpath $B$ being merged with subpath $A$. The the last case we can move $v$ to the start of $B$ in $i \cdot \pairshortcut{G}$ hops, and then to the start of $A$ in another $\pairshortcut{G}$, for a cumulative $(i+1)\pairshortcut{G}$ hops, proving the shortcut dilation claim since $i = O(\log n)$. This completes the proof of subclaim 1.

  Before we proceed, we remind the reader about the \emph{heavy-light decomposition} of a tree, this time going into more detail. Consider a (rooted) tree with at most $n$ nodes. We refer to the number of nodes that are in the subtree of $v$ as the subtree size of $v$. For each node $v$ we label the edge connecting $v$ to its child node of larger subtree size as \emph{light}, and the remaining edges as \emph{heavy}. It is easy to show that each root-leaf path contains $O(\log n)$ heavy edges and that the light edges form a collection of vertex-disjoint paths. 
  
  \textbf{Subclaim 2:} Suppose that $f : E \to \{\to, \leftarrow, \bot\}$ is such that $G(f)$ is precisely a union of (maximal) vertex-disjoint rooted \textbf{trees} $P_1 \cup P_2 \ldots \cup P_q$. We argue there is a distributed $\tilde{O}(T)$-round algorithm that takes the distributed knowledge of $f$ as input and (distributedly) learns for each node $v$ that is in the (maximal) tree $P_i$ (i) the unique tree ID of $P_i$, the root of $P_i$, and the parent of its root in $G(f)$ if any, (ii) the subtree size with respect to $P_i$, (iii) the (distributed) heavy-light edge labelling of $P_i$, (iv) the number of heavy edges on the path between the root of $P_i$ and $v$.

  \textbf{Proof of subclaim 2:} We first note that, given vertex-disjoint clusters (the clusters are subtrees of $P_i$) where values (i)--(iv) are maintained, one can easily construct shortcuts on each cluster (and perform partwise aggregation on it). First, we construct shortcut via subclaim 1 on all light paths of all clusters combined: we ignore all non-light edges and note that light edges can be partitioned into vertex-disjoint paths, hence the preconditions of subclaim 1 are valid. Then the heavy edges of each cluster are added to this shortcut of that cluster. Since every root-leaf path has at most $\tilde{O}(1)$ heavy edges and the set of all light paths have $\tilde{O}(\pairshortcut{G})$-quality shortcuts, this shortcut can clearly be argued to have $\tilde{O}(\pairshortcut{G})$ quality with respect to the set of all clusters.
  
  Furthermore, given a (distributed) node function $g : V \to [n^{O(1)}]$ and a set of vertex-disjoint clusters (i.e., the clusters are subtrees of $P_i$), for each node $u$ we can calculate the sum of the values $g(v)$ for all nodes $v$ in the subtree of $u$ with respect to the cluster $u$ is in. First, using subclaim 1, we partition all clusters into vertex-disjoint light paths and calculate for each node $v$ the sum of all values $g$ from nodes reachable from $v$ in its light path. Then, for $O(\log n)$ step we calculate the sum of $g$'s on the entire path and forward this sum to the parent path, which then disseminates this value via the shortcut of subclaim 1 to all of its nodes, increasing subtree sum. After $O(\log n)$ step, all light paths have the correct value since any root-to-leaf path in the cluster has $O(\log n)$ heavy edges.

  Finally, we prove the subclaim in a similar manner to subclaim 1. The algorithm starts with initially singleton clusters and then grows the clusters for $O(\log n)$ steps. In the first step, each node is its own cluster. In each subsequent step we assume we have maintained the above data (i)--(iv) for each node $v$ with respect to the cluster (i.e., subtree) that $v$ is in. We show how to merge neighboring clusters while maintaining the validity of this data. First, each cluster flips an independent random heads/tails coin (the entire cluster flips a common coin without communication via shared randomness). Intuitively, tails will merge into heads and only if ``heads subtree'' has a directed edge towards the ``tails subtree''. The merging is fairly simple. Consider a head subtree with its incident tail subtrees (possibly multiple ones). We can easily update the values (i) by construct a shortcut for each cluster via the above remark and disseminating the appropriate data, like in Subclaim 1. To compute (ii), we pass from the (root of the) tail subtrees the sizes of these trees to the root's parent node in the head subtree. Suppose that each node $v$ in the head subtree receives a cumulative sum of values $g(v) \ge 0$ via this process. We set up $g^\dag(v) \gets g(v) + 1$ and calculate for each node $u$ the sum of $g^\dag$'s in the cluster's subtree of $u$ via the above remark. This value is exactly the updated subtree size of a node $u$ with respect to its updated cluster. The value (iii) can be directly calculated from the subtree sizes we computed in (ii). And finally, (iv) can be computed in an analogous way as (ii). This completes the proof of the subclaim.

  \textbf{Completing the proof.} Fix a (distributed) vertex partition $\calP = (P_1, \ldots, P_k)$. The algorithm starts with initially singleton clusters and then grows the clusters for $O(\log n)$ steps. In the first step, each node is its own cluster. Each cluster is a subset of some part $P_i$ and in the end the the cluster will be equal to the entire $P_i$. In each step there is some partition into clusters and for that step we maintain a distributed function $f : E \to \{\to, \leftarrow, \bot\}$ such that $G(f)$ consists of maximal rooted spanning trees of the clusters (in other words, each cluster has a unique spanning tree in $G(f)$ and each maximal spanning tree in $G(f)$ determines a cluster). First, each cluster flips an independent random heads/tails coin (the entire cluster flips a common coin without communication via shared randomness). Intuitively, tails will merge into heads and only if ``heads cluster'' has a directed edge towards the ``tails cluster''. The only issue we encounter is that the ``orientation'' of the head cluster spanning tree might not match the orientation of the tail cluster spanning tree. To this end, we need to ``fix'' each tail cluster $B$ in the following way. Suppose that the head cluster $A$ has a directed edge from some node in $A$ to a node $b$ in $B$. We need to reverse the orientation of edges on the root (of $B$)-to-$b$ path. However, this can be easily done via Subclaim 2 (note that the tail components in isolation satisfy subclaim 2 and hence we can invoke the subclaim upon them). After that we add the directed edge between $A$ and $B$ to $f$ and we are done, the merged components are correctly oriented and satisfy the required invariant for $f$. After $O(\log n)$ rounds of merging the procedure completes and clusters exactly correspond to $P_1, P_2, \ldots, P_k$ with high probability. At this moment we use subclaim 2 to build a shortcut on $\{P_1, \ldots, P_k\}$ and we are done.
%
\end{proof}

\bibliographystyle{acmsmall}
\bibliography{abb,refs,ref}

\end{document}